\documentclass[draftclsnofoot,onecolumn]{IEEEtran}

\ifCLASSINFOpdf
\else
\fi
\usepackage[pdftex]{graphicx}
\usepackage{amsmath,amssymb,amsthm}
\usepackage{subfig}
\usepackage{bm}
\usepackage{algorithm}
\usepackage{algpseudocode}
\usepackage{changepage}

\newtheorem{theorem}{Theorem}
\newtheorem{lemma}{Lemma}

\theoremstyle{definition}

\begin{document}
%
\title{Storage and Repair Bandwidth Tradeoff for Distributed Storage Systems with Clusters and Separate Nodes}




%
\author{\IEEEauthorblockN{Jingzhao Wang\IEEEauthorrefmark{1},
Tinghan Wang\IEEEauthorrefmark{2}
Yuan Luo\IEEEauthorrefmark{3}}\\
\IEEEauthorblockA{\IEEEauthorrefmark{1}Email: wangzhe.90@sjtu.edu.cn}
\IEEEauthorblockA{\IEEEauthorrefmark{2}Email: wth19941018@sjtu.edu.cn}
\IEEEauthorblockA{\IEEEauthorrefmark{3}Email: yuanluo@sjtu.edu.cn}\\
Department of Computer Science and Engineering\\ Shanghai Jiao Tong University, Shanghai 200240, China}


\maketitle

\begin{abstract}
The optimal tradeoff between node storage and repair bandwidth is an important issue for distributed storage systems (DSSs). As for realistic DSSs with clusters, when repairing a failed node, it is more efficient to download more data from intra-cluster nodes than from cross-cluster nodes. Therefore, it is meaningful to differentiate the repair bandwidth from intra-cluster and cross-cluster. For cluster DSSs the tradeoff has been considered with special repair assumptions where all the alive nodes are utilized to repair a failed node. In this paper, we investigate the optimal tradeoff for cluster DSSs under more general storage/repair parameters. Furthermore, a regenerating code construction strategy achieving the points in the optimal tradeoff curve is proposed for cluster DSSs with specific parameters as a numerical example. Moreover, the influence of separate nodes for the tradeoff is also considered for DSSs with clusters and separated nodes.
\end{abstract}


%
\IEEEpeerreviewmaketitle

\section{Introduction}
As data center storage expands at scale, storage node failures are more prevalent \cite{Jiang2008}, where distributed storage systems (DSSs) with erasure coding are widely utilized to ensure data reliability \cite{Zhang2017Tree, Li2017Beehive,Huang2013}. When a storage node in DSSs has failed, to recover the failed node, a new node will download data from others which are called helper nodes. The amount of data to download is called the repair bandwidth. In \cite{Dimakis2010}, the tradeoff between node storage and bandwidth to repair one node is investigated for homogeneous DSSs \cite{Ernvall2013} where all the nodes (hard disks or other storage devices) have the same parameters (storage per node, repair bandwidth, etc.). Meanwhile, regenerating codes are proposed based on the tradeoff to reduce the repair bandwidth of DSSs.

Contrast to homogeneous DSSs, in heterogeneous DSSs~\cite{Ernvall2013,yu2015tradeoff}, nodes can have different storage and repair bandwidths. In~\cite{Akhlaghi2010}, the communication cost among nodes is taken into consideration, where the storage of each node is equal, but the repair bandwidth varies based on the location of failed nodes. In realistic storage systems, nodes in the same cluster (rack) may be connected to each other with cheaper and faster networks (i.e. local area networks)~\cite{ford2010}, where downloading data from each other may be faster and cheaper. For the sake of reducing communication cost, it is more efficient to download more data from intra-cluster nodes and less from cross-cluster nodes, rather than downloading the same amount of data from others.

This paper investigates a type of heterogenous DSSs with clusters and separate nodes, where the tradeoff between node storage and repair bandwidth is characterized on more flexible parameter settings. The tradeoff for DSSs with only two clusters (racks) is analysed in \cite{Two_rack2013}. In \cite{Prakash2017}, the authors consider cluster DSSs with one relay node in each cluster. The relay node collects data in its cluster and transmits to nodes in other clusters. In \cite{sohn2017capacity}, the properties of DSS with multiple clusters are considered under a specific assumption that the new node will download data from all the other nodes when one node has failed. In current paper, the optimal tradeoff for cluster DSSs is investigated under more general settings that the new node does not need to download data from all the other nodes. The traditional homogeneous DSS and model in \cite{Two_rack2013} and \cite{sohn2017capacity} can be obtained by specializing parameters of our general model.  On the other hand, the tradeoff for DSSs with clusters and separate nodes is also analysed. Moreover, a regenerating code construction strategy is investigated for cluster DSSs with specific parameters.

The rest of this paper is organized as follows. The model of DSS with clusters and separate node (CSN-DSS) is introduced and the problem is formulated in Section~\ref{sec_Pre}. In Section~\ref{sec_Analysis}, the properties of cluster DSSs are analysed in two parts, which are proved in Theorem~\ref{theorem_MC} and Theorem~\ref{theorem_MC2}, respectively. As a general case, the DSS with clusters and one separate node is analysed in Theorem~\ref{theorem_MC3}. Afterward, the tradeoff between node storage and repair bandwidth is characterized. Some numerical results are illustrated in Section~\ref{sec_consruction}, where a regenerating code construction strategy is investigated for cluster DSSs. Finally, Section~\ref{sec_conclusion} presents the conclusion and future work.

\section{Preliminaries and Problem Formulation}\label{sec_Pre}
Subsection~\ref{subsec_CSND} defines main parameters of CSN-DSSs. As an efficient tool to analyse DSSs, the information flow graph is introduced in Subsection~\ref{subsec_IFG}. The problem investigated in this paper is formulated in Subsection~\ref{subsec_PM}.

\subsection{Distributed Storage System with Cluster and Separate Nodes (CSN-DSS)}\label{subsec_CSND}

\begin{figure*}
\begin{minipage}{0.25\textwidth}
    \centering
    \includegraphics[width=0.9\textwidth]{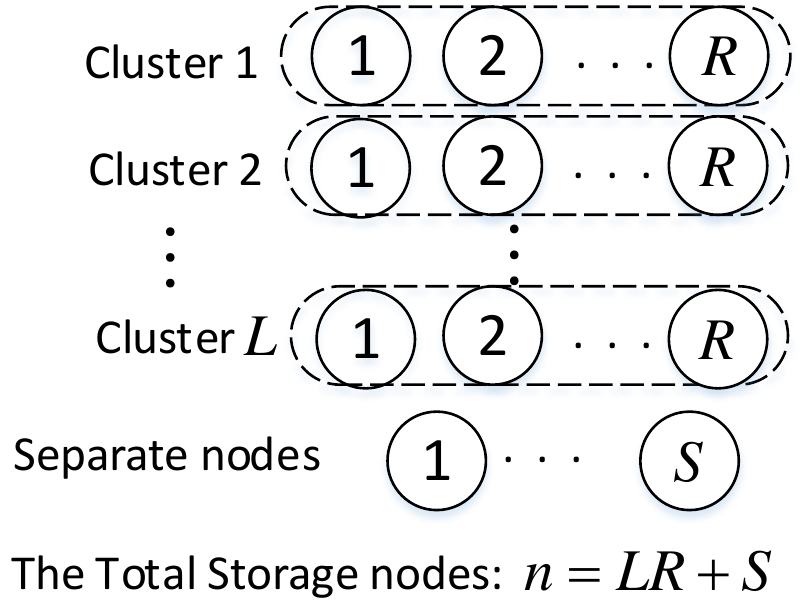}
    \caption{System model for CSN-DSS}\label{fig_CN}
\end{minipage}
\hspace{0.01\textwidth}
\begin{minipage}{0.74\textwidth}
    \subfloat[One node in cluster 1 and one separate node have failed]{
    \begin{minipage}[t]{0.5\textwidth}
    \centering
    \includegraphics[width=\textwidth]{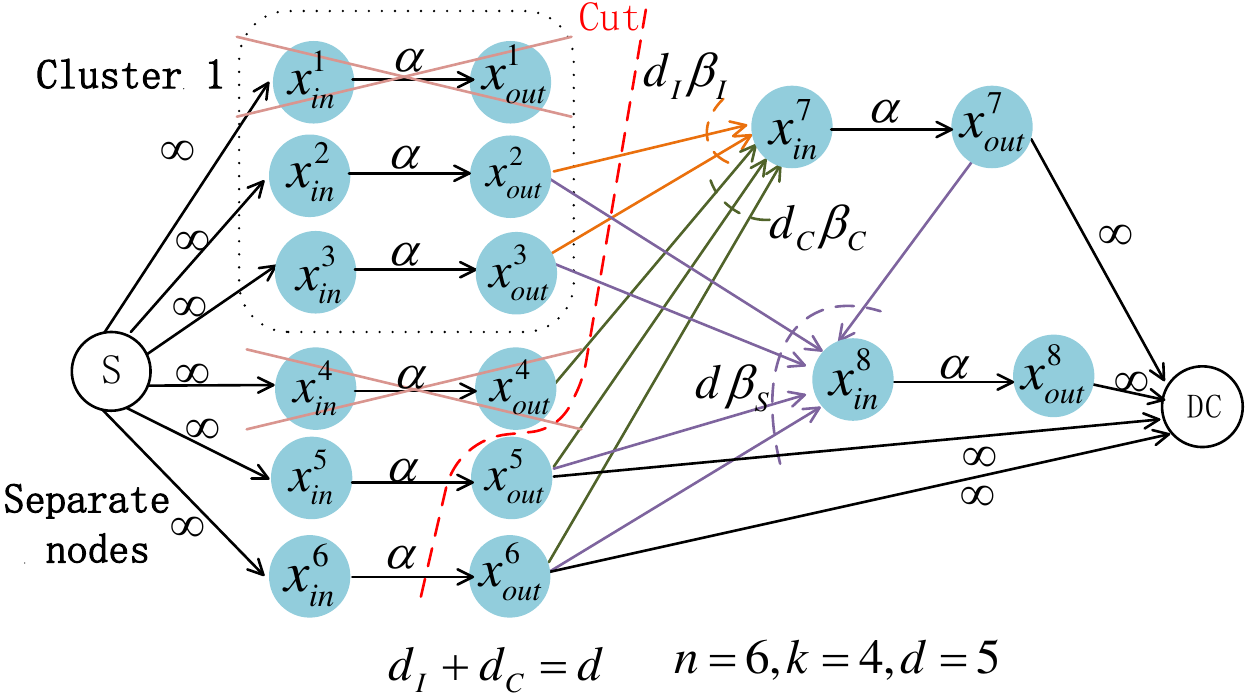}
    \end{minipage}}
    \subfloat[Two nodes in cluster 1 have failed]{
    \begin{minipage}[t]{0.5\textwidth}
    \centering
    \includegraphics[width=\textwidth]{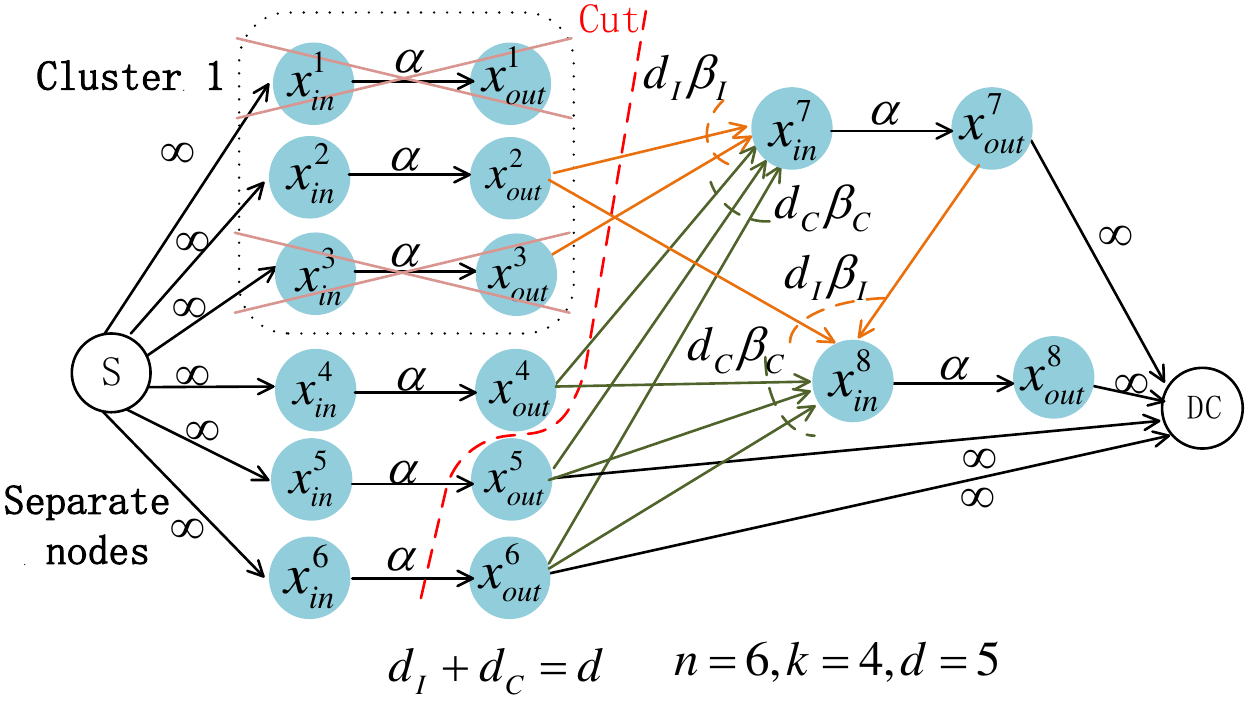}
    \end{minipage}}
    \caption{The IFGs of one cluster and separate nodes distributed storage system}\label{fig_IFGCSN}
    \end{minipage}
\end{figure*}

\textbf{CSN-DSS Model:} The cluster and separate nodes distributed storage system, illustrated in Figure~\ref{fig_CN}, consists of $S$ separate nodes and $L$ clusters each with $R$ nodes. The total number of the storage nodes is $n=LR+S$. A data file of size $\mathcal{M}$ symbols is divided into $k$ fragments, each of size $\alpha=\mathcal{M}/k$ symbols. The $k$ fragments are encoded into $n$ fragments of size $\alpha$ and stored at $n$ nodes. Any $k$ encoded fragments out of $n$ suffice to recover the original data file, which is called the $(n,k)$ MDS\footnote{An $(n,k)$ maximum distance separate (MDS) code encodes $k$ information symbols to $n$ symbols such that any $k$ symbols of $n$ suffice to recover the original information symbols} property.

When repairing a failed node, a newcomer locating in the same cluster will download data from other alive nodes and repair the failed one. If the failed node is a separate node, the newcomer is still a separate node.

If \textbf{a node in cluster} has failed, the newcomer downloads $\beta_I$ symbols each from $d_I$ intra-cluster nodes and $\beta_C$ symbols each from $d_C$ cross-cluster nodes (including nodes from other clusters and separate nodes). Let
\begin{equation*}\label{equ_d}
  d\triangleq d_I+d_C
\end{equation*}
and $d\geq k$ based on the $(n,k)$ MDS property \cite{Survey2011}.

In realistic distributed storage systems, the storage servers may connect to each other with local area networks or external networks. The communication between servers in the same local area network is cheaper than servers in different local area networks and connected with external networks. The servers connected in the same local area network can be seen as being in the same cluster. The separate nodes connect to cluster nodes by external networks. In order to reduce the bandwidth cost, it is better to download more data from nodes in the same cluster and less data from outer nodes, namely, $\beta_I\geq \beta_C$. On the other hand, the intra-cluster nodes are used preferentially in general case, when repairing failed nodes. Therefore, it is reasonable to assume that all the intra-cluster alive nodes are used for repair, namely $d_I=R-1$ in current paper.

If \textbf{a separate node} has failed, the newcomer downloads $\beta_S$ symbols each from $d$ other nodes including nodes in clusters and separate nodes, which means that the newcomer needs only $d$ helper nodes to repair a failed node, no matter where the failed node is located.

In a CSN-DSS, we dub $(n,k,L,R,S)$ the \textbf{node parameters} and $(\alpha, d_I, \beta_I, d_C, \beta_C, \beta_S)$ the \textbf{storage/repair parameters} for simplicity. The intra-cluster and cross-cluster bandwidth of repairing a cluster node is defined as
\begin{small}$$\gamma_I\triangleq d_I\beta_I \text{ \ and\ }\gamma_C \triangleq d_C\beta_C,$$\end{small} respectively.
The bandwidth of repairing a separate node is $\gamma_S\triangleq d\beta_S.$ When $\beta_I=\beta_C=\beta_S$, the traditional homogeneous DSS in \cite{Dimakis2010} is obtained.

\subsection{Information Flow Graph (IFG)}\label{subsec_IFG}

To analyse the performance of distributed storage systems, the information flow graph is proposed in~\cite{Dimakis2010}, which consists of three kinds of nodes: a single data source $S$, storage nodes $x_{in}^i$, $x_{out}^i$, and a data collector $DC$ as shown in Figure~\ref{fig_IFGCSN} (a). A physical storage node (hard disk or other storage device) is represented by a storage input node $x_{in}^i$ and an output node $x_{out}^i$, where pre-computing is permitted when transmitting data. $x_{in}^i$ and $x_{out}^i$ are connected by a directed edge with capacity identical to the storage size $\alpha$ of the node. Throughout this paper, $x^i$ is used to present $x_{in}^i$ and $x_{out}^i$ as a storage node.

As is mentioned before, the original data file is divided into $k$ fragments and encoded into $n$ fragments stored at $n$ nodes, which is represented by $n$ edges from node $S$ to $\{x_{in}^i\}_{i=1}^n$ with infinite capacity.

At the initial time, the source node $S$ stores data to the $n$ storage nodes. Then $S$ becomes inactive and the $n$ storage nodes become active. When a node $x^j$ fails, it becomes inactive. The repair procedure creates an active newcomer $x^{n+1}$ to the graph as a substitute by connecting edges each with capacity $\beta$ from $d(\geq k)$ surviving active nodes, which means the newcomer downloads $\beta$ symbols from each of $d$ alive nodes. Note that he values of $\beta$ vary in heterogenous DSSs. The total number of active nodes remains $n$ after each repair. Based on the $(n,k)$ MDS, a data collector $DC$ connects to arbitrary $k$ active nodes with direct edges of infinite capacity, which means any $k$ nodes suffice to reconstruct the original data file. There may be many $DC$s connecting different sets of $k$ active nodes, but only one $DC$ is drawn visually.

An example of information flow graph for CSN-DSSs is illustrated in Figure~\ref{fig_IFGCSN} (a), where $n=6(L=1,R=3,S=1),k=4,d=5$. Two node $x^1$ and $x^6$ have failed successively. When node $x^1$ (a node in cluster 1) has failed, a newcomer $x^7$ downloads $\beta_I$ symbols each from $x^2$ and $x^3$ (two intra-cluster nodes) and $\beta_C$ symbols each from $x^4$, $x^5$ and $x^6$ (three cross-cluster nodes). When node $x^6$ has failed, a second newcomer $x^8$ is added by downloading $\beta_S$ symbols from five alive nodes (nodes in clusters or separate nodes). This model only handles one node failure at a time, downloading data from $d$ helper nodes, a subset of the $n-1$ alive nodes.

An important notion associated with the information flow graph is that of minimum cuts: In an IFG, a (direct) cut between $S$ and $DC$ is defined as a subset $\mathcal{C}$ of edges such that every directed path from $S$ to $DC$ contains at least one edge in $\mathcal{C}$. The min-cut is the cut between $S$ and $DC$ in which the total sum of the edge capacities is smallest.

\subsection{Problem Formulation}\label{subsec_PM}

As is proved in~\cite{Dimakis2010}, the reconstruction problem for every $DC$ reduces exactly to multicasting the original data from a single source $S$ to every $DC$. With relative works on network coding~\cite{Li2003, Ahlswede2006}, a tradeoff between node storage and repair bandwidth can be maintained by analysing the min-cuts between $S$ and all possible $DC$s.

Let $\mathcal{G}$ be the set of all possible information flow graphs of a CSN-DSS with node parameters $(n,k,L,R,S)$ and storage/repair parameters $(\alpha, d_I, \beta_I, d_C, \beta_C, \beta_S)$. Consider any given finite IFG $G\in \mathcal{G}$, with a finite set of data collectors. If the minimum of the min-cuts separating the source with each data collector is larger than or equal to the data object size $\mathcal{M}$, then there exists a linear network code such that all data collectors can recover the data object (see Proposition 1 in~\cite{Dimakis2010}). Denote the graph with minimum min-cut by $G^*$. The \textbf{capacity} of a CSN-DSS is defined as
\begin{small}\begin{eqnarray*}
\mathbf{C}(\mathcal{G})\triangleq \text{min-cut of }G^*.
\end{eqnarray*}\end{small}
In order to send data of size $\mathcal{M}$ from the source to any data collectors,
\begin{small}\begin{equation}\label{equ_bound}
\mathbf{C}(\mathcal{G})\geq\mathcal{M}
\end{equation}\end{small}
should be satisfied. As $\mathbf{C}(\mathcal{G})$ depends on node parameters and storage/repair parameters, when node parameters $(n,k,L,R,S)$ are fixed, a tradeoff between node storage $\alpha$ and repair bandwidth parameters $(d_I, \beta_I, d_C, \beta_C, \beta_S)$ will be characterized. The set of points $(\alpha, d_I, \beta_I, d_C, \beta_C, \beta_S)$ which satisfies $\mathbf{C}(\mathcal{G})\geq\mathcal{M}$ is feasible in the sense of reliably storing the original file of size $\mathcal{M}$.

Therefore, to analyse the capacity or tradeoff properties of CSN-DSSs is to analyse the min-cuts of IFGs for given node parameters, which is illuminated in the following sections.

\section{Analysis of CSN-DSSs}\label{sec_Analysis}

In this section, we investigate the min-cuts of IFGs for CSN-DSSs and prove the algorithms to generate the IFG achieving the capacity of given CSN-DSS. Some useful terms and notations are defined in Subsection~\ref{subsec_Mincut}. Subsection~\ref{subsec_NoSN} considers the min-cuts of IFGs with no separate selected nodes. With similar methods, the influence of one separate selected node is investigated in Subsection~\ref{subsec_WithSN}.

\subsection{Terminologies and Min-cut Calculation} \label{subsec_Mincut}
 For a given IFG $G \in \mathcal{G}$, the main problem is to find the min-cuts between source $S$ and each $DC$. Because there are no paths among different $DC$s, the min-cuts between $S$ and each different $DC$ can be analysed in the same way. For simplicity, assume the IFGs in the following parts only contain one single $DC$ and the min-cut only indicates the min-cut separating $S$ and $DC$. This subsection introduces some important terminologies such as repair sequence, selected node distribution, cluster order and relative location, with which the method for calculating min-cuts of IFGs is illuminated.

\textbf{Topological order:} Note that every directed acyclic graph has a topological order (see~\cite{Networkflows1993}, Chapter 3), which is an ordering of its vertices such that the existence of a path from $v_i$ to $v_j$ implies $i<j$. The $k$ output nodes connected by every $DC$ can be topologically sorted.

\textbf{Min-cut between $S$ and $DC$:} Let $\{x_{out}^{t_i}\}_{i=1}^k$ be the set of output nodes connected by the data collector, which are topologically ordered. The min-cut between $S$ and $DC$ can be calculated out by cutting $\{x_{out}^{t_i}\}_{i=1}^k$ one by one in the topological order, which is proved in~\cite{Dimakis2010} Lemma 2. Each time cutting a node, a part of the min-cut is determined, called a \textbf{part-cut value}. So the min-cut between $S$ and $DC$ is the summation of the $k$ part-cut values. For example, in Figure~\ref{fig_IFGCSN} (a), The $DC$ connects to four output nodes $x_{out}^6, x_{out}^5, x_{out}^7, x_{out}^8$, which are topologically ordered. Cut the five nodes one by one, as is shown by the red dashed line, we then get the four part-cut values $\alpha, \alpha, (\beta_C+2\beta_I)(\leq \alpha), (2\beta_C)(\leq \alpha)$ respectively. Note that if $\beta_C+2\beta_I \geq \alpha$, the cut line will be between $x_{in}^7$ and $x_{out}^7$. As a result, the third part-cut value will change to $\alpha$.

\textbf{Repair sequence and selected nodes}: It is obvious that when a $DC$ connects to a newcomer instead of connecting to an original node, the part-cut value may be smaller than $\alpha$. So smaller min-cuts can be derived as the $DC$ connects to more newcomers. Based on the MDS property mentioned in Subsection \ref{subsec_CSND},  a $DC$ connects to $k(\leq n)$ nodes, it is possible to find an IFG with a $DC$ only connecting to $k$ newcomers. Note that each newcomer corresponds to an original node failure and completes the repair procedure. Clearly, the topological order of $k$ output nodes $\{x_{out}^{t_i}\}_{i=1}^k$ corresponds to a \textbf{repair sequence} of original nodes. These original nodes contained in a repair sequence are called \textbf{selected nodes}.

In homogeneous distributed storage systems, all the storage/repair parameters ($\alpha, \beta, d$) are the same for different nodes. The repair sequence won't affect the minimum min-cut.
As is proved in~\cite{Dimakis2010}, when the $DC$ connects to $k$ newcomers and the newcomer $x^{t_i}$ downloads data from all the former newcomers $\{x^{t_j}\}_{j=1}^{i-1}$, the minimum min-cut is reached.

However, in a CSN-DSS, the storage/repair parameters are different for nodes in cluster and separate nodes. Different repair sequences result in different min-cuts of the IFGs. As is illustrated in Figure~\ref{fig_IFGCSN}, there two different repair sequences ($x^1$, $x^4$) in (a) and ($x^1$, $x^3$) in (b). The corresponding min-cuts are differently $2\alpha+2\beta_I+\beta_C$ and $2\alpha+2\beta_I+\beta_I$ respectively, as are shown by the red dash cut lines. Here we only consider two newcomer for the simplicity of the figures and assume $2\beta_I, \beta_C$ are less than $\alpha$. Consequently, the repair sequence determines the minimum min-cut directly.

\textbf{Selected node distribution:} For any given $k$ selected nodes, without loss of generality, assume the clusters are relabeled by the number of selected nodes in descending order. In the other words, cluster $1$ contains the most selected nodes, and cluster $L$ contains the least selected nodes. Define the \textbf{selected node distribution} as $\textbf{s}=(s_0,s_1,s_2,...,s_L)$, where $s_i(1\leq i\leq L)$ is the number of selected nodes in cluster $i$, and the first component $s_0$ is the number of selected separate nodes. Meanwhile, the set of all possible selected node distributions is defined as follows.
\begin{small}
\begin{equation*}
\mathcal{S}=\left\{\textbf{s}=(s_0,s_1,s_2,...,s_L):\ s_{i+1}\leq s_{i},\ 0\leq s_i\leq R,\ \text{for}\ 1\leq i\leq L;\ 0\leq s_0\leq S;\ \sum_{i=0}^L s_i=k\right\}
\end{equation*}
\end{small}
Note that the selected node distribution describes the total number of selected nodes in each cluster and separate nodes. Moreover, we need to represent the topological order of $k$ selected output nodes $\{x_{out}^{t_i}\}_{i=1}^k$ corresponding to $k$ selected original nodes, called cluster order.

\textbf{Cluster order:} Let the \textbf{cluster order} $\bm{\pi}=(\pi_1,\pi_2,...,\pi_k)$ denote the repair sequence, where $\pi_i(1\leq i\leq k)$ is the index of the cluster which contains the newcomer $x^{t_i}$ corresponding to the failed node. If the $i$th node is a separate node,
$\pi_i$ equals to $0$. Note that the cluster index is enough to define the repair sequence, because the storage/repair parameters for each node in the same cluster are the same.

For a certain selected node distribution $\textbf{s}=(s_0, s_1, s_2,..., s_L)$, there are different cluster orders. The set of possible cluster orders is defined as
\begin{small}
\begin{equation*}
\Pi(\textbf{s})=\Big\{\bm{\pi}=(\pi_1,..., \pi_k):\ \sum_{j=1}^k \mathbb{I}(\pi_j= i)=s_i,\  i\in\{0,1,...,L\}\Big\},
\end{equation*}
\end{small}
where $\mathbb{I}(\pi_j=i)$ is an indicator function which equals $1$ if $\pi_j=i$, and $0$ otherwise.

The relationship between selected node distribution and cluster order is illustrated in Figure~\ref{fig_sequence}, where the selected nodes are numbered. The selected node distribution $\textbf{s}=(1,4,3,1)$ means that, in the IFG of this CSN-DSS, the $DC$ connects $1$ separate node, $4$ nodes from cluster 1, $3$ nodes from cluster 2 and $1$ node from cluster 3. The cluster order $\bm{\pi}(\textbf{s})=(1,2,3,1,2,1,2,1,0)$ is a possible repair sequence for $\textbf{s}$.

The selected nodes are labeled from $1$ to $k$ as Figure~\ref{fig_sequence} shows, although it's enough to record the cluster number in the cluster order as the nodes in one cluster are undifferentiated. For the nodes in a cluster order $\bm{\pi}$, it's also needed to identify the precedence of selected nodes in each cluster. Assume the $i$-th node in cluster order $\bm{\pi}$ is the $h_{\bm{\pi}}(i)$-th node in its cluster. We called $h_{\bm{\pi}}(i)$ the \textbf{relative location} of the $i$-th node and
\begin{equation}\label{equ_hi}
h_{\bm{\pi}}(i)=\sum_{j=1}^i \mathbb{I}(\pi_j=\pi_i),
\end{equation}
where $1\leq i\leq k$. For an example, in Figure~\ref{fig_sequence}, the cluster order is $\bm{\pi}=(1,2,3,1,2,1,2,1,0)$ and $h_{\bm{\pi}}(4)=\mathbb{I}(\pi_1=\pi_4)+\mathbb{I}(\pi_2=\pi_4)+\mathbb{I}(\pi_3=\pi_4)+\mathbb{I}(\pi_4=\pi_4)=1+0+0+1=2$ where $\pi_4=1$. The corresponding sequence of $h_{\bm{\pi}}(i)$ is then $(1,1,1,2,2,3,3,4,1)$.

\begin{figure}[t]
  \centering
  \begin{minipage}[t]{0.50\textwidth}
    \centering
    \includegraphics[width=0.6\textwidth]{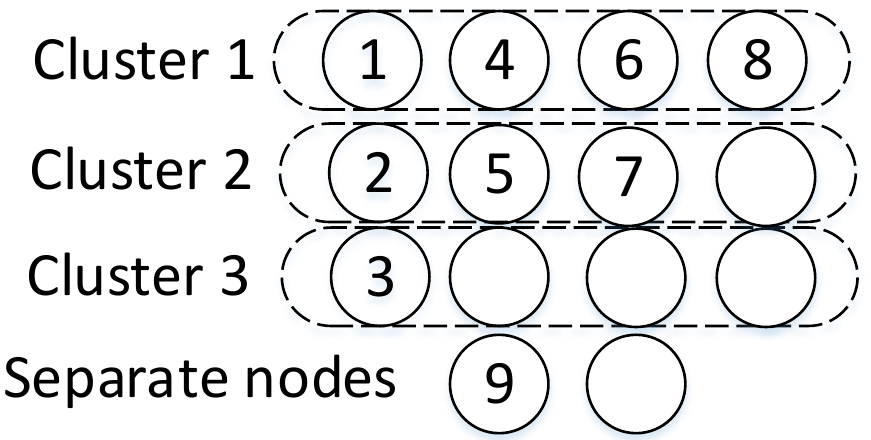}
    \caption{The numbered nodes are selected nodes and the selected node distribution is $\textbf{s}=(1, 4, 3, 1)$. A corresponding cluster order is $\bm{\pi}=(1,2,3,1,2,1,2,1,0)$ for the CSN-DSS with $n=15$ nodes, $k=9$ selected nodes.} \label{fig_sequence}
  \end{minipage}
  \hspace{0.08\textwidth}
  \begin{minipage}[t]{0.40\textwidth}
    \centering
    \includegraphics[width=0.7\textwidth]{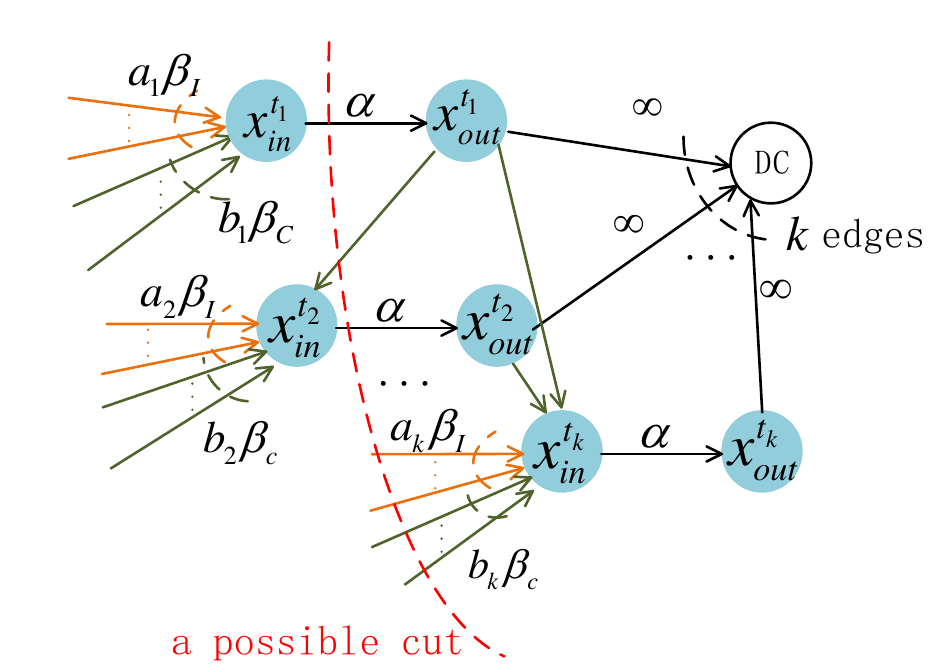}
    \caption{The min-cut IFG $G$ with no separate nodes in selected nodes} \label{fig_partcut}
  \end{minipage}
\end{figure}

\textbf{Calculating the min-cut between $S$ and $DC$ of an IFG:}

When calculating the min-cut of IFG $G$, it is to cut the output nodes contacted by the $DC$ $k$ times in topological order. Consider two disjoint sets $U$ and $\overline{U}$ of the nodes in $G$. Assume $S$ and the original nodes $x^i(1\leq i\leq n)$ are contained in $U$ and $DC$ is contained in $\overline{U}$ at the beginning. Every time we cut $G$, some nodes are added into $U$ and $\overline{U}$ respectively. When $G^*$ is cut $k$ times, all the nodes of $G^*$ are contained in $U$ or $\overline{U}$ and the set of edges emanating from $U$ to $\overline{U}$ is a cut between $S$ and $DC$. Let $\mathcal{C}$ denote the edges in the cut set, i.e., the set of edges going from $U$ to $\overline{U}$.

As is illustrated in Figure~\ref{fig_partcut}, the $k$ selected nodes $\{x^{t_i}\}_{i=1}^k$ are in topological order. When cutting node $x^{t_1}$, there are two possible cases.
\begin{adjustwidth}{0.5cm}{0cm}
  $\bullet$ If $x_{in}^{t_1}$ is in $U$, the edge $(x_{in}^{t_1}, x_{out}^{t_1})$ is contained in $\mathcal{C}$. The part-cut value equals $\alpha$.\\
  $\bullet$ In case of $x_{in}^{t_1}$ is in $\overline{U}$, since $x_{in}^{t_1}$ has an in-degree of $d=d_I+d_C$ and it is the topologically first newcomer in $\overline{U}$, all the incoming edges of $x_{in}^{t_1}$ must be in $\mathcal{C}$, which consists of $d_I$ edges from intra-cluster nodes and $d_C$ edges from cross-cluster nodes. The part-cut value equals $d_I\beta_I+d_C\beta_C$.
\end{adjustwidth}
\noindent When cutting node $x^{t_2}$, the first case is similar to $x^{t_1}$ and the part-cut value is also $\alpha$. For the second case, if $x^{t_2}$ is in the same cluster with $x^{t_1}$, the incoming edges of $x_{in}^{t_2}$ consist of $d_I-1$ edges from intra-cluster nodes and $d_C$ edges from cross-cluster nodes, then the part-cut value equals to $(d_I-1)\beta_I+d_C\beta_C$. On the other hand, if $x^{t_2}$ is in different clusters with $x^{t_1}$, the incoming edges of $x_{in}^{t_2}$ still contain $d_I$ edges from intra-cluster nodes but $d_C-1$ edges from cross-cluster nodes, and the part-cut value equals $d_I\beta_I+(d_C-1)\beta_C$.

\noindent Now consider node $x^{t_i}(1\leq i\leq k)$, the $i$th newcomer:
\begin{adjustwidth}{0.5cm}{0cm}
  $\bullet$ If $x_{in}^{t_i}\in U$, the edge $(x_{in}^{t_i},x_{out}^{t_i})$ must be in $\mathcal{C}$.\\
  $\bullet$ If $x_{in}^{t_i}\in \overline{U}$, node $x_{in}^{t_i}$ has $d=d_I+d_C$ incoming edges consisting of two parts: edges from nodes in $U$ and edges from nodes in $\overline{U}$. Cut set $\mathcal{C}$ only includes the first part of incoming edges among which let $a_i$ denote the number of \textbf{edges from intra-cluster nodes} and $b_i$ denote \textbf{edges from cross-cluster nodes}. It's obvious that $0\leq a_i\leq d_I$ and $0\leq b_i\leq d_C$. Note that when $i$ increases by $1$, either $a_i$ or $b_i$ will decrease by $1$ and will not decrease when $a_i$ or $b_i$ equals $0$. Since at most $i-1$ incoming edges of $x_{in}^{t_i}$ can be from $x_{out}^{t_j}(1\leq j\leq i-1)$ already contained in $\overline{U}$,
\begin{equation}
a_i+b_i \geq d-(i-1), \label{equ_bi}
\end{equation}
for $1\leq i\leq k$. Equality holds if $a_i$ and $b_i$ will not decrease to $0$ as $i$ increases.

\noindent If the $i$th selected node is a separate node, let $c_i$ denote the number of incoming edges of $x_{in}^{t_i}$ from $U$ and
\end{adjustwidth}

\begin{equation*}\label{equ_ci}
c_i=d-(i-1).
\end{equation*}
The respective values of $a_i$ and $b_i$ depend on the repair sequence of original nodes, namely, the selected node distribution $\textbf{s}$ and cluster order $\bm{\pi}$. The sum of the capacity of these edges is called the $i$th \textbf{part incoming weight}

\begin{equation} \label{equ_wi}
w_i(\textbf{s},\bm{\pi})=\begin{cases}
               a_i\beta_I+b_i\beta_C \ &\text{if the }i\text{th selected node is a cluster node},\\
               c_i\beta_S\ &\text{if the }i\text{th selected node is a separate node}
               \end{cases}.
\end{equation}
If the selected node distribution $\textbf{s}$ or cluster order $\bm{\pi}$ is fixed beforehand,  $w_i(\textbf{s},\bm{\pi})$ can be written as $w_i(\bm{\pi})$ or $w_i$ for simplicity. On the other hand, $a_i, b_i, c_i$ can be written as $a_i(\bm{\pi}), b_i({\bm{\pi}}), c_i({\bm{\pi}})$ for specific $\bm{\pi}$, respectively.

For a fixed selected node distribution $\textbf{s}$, the min-cut varies for different cluster orders $\bm{\pi}\in \Pi(\textbf{s})$. The min-cut for $\bm{\pi}=(\pi_1,\pi_2,...,\pi_k)$ is defined as

\begin{equation}\label{equ_MC}
MC(\textbf{s},\bm{\pi})\triangleq\sum_{i=1}^k\min\{w_i(\bm{\pi}),\alpha\}.
\end{equation}
With the above definitions, for an IFG $G$ with specified selected node distribution and cluster order, the min-cut can be figured out. In the following subsection, the min-cuts of cluster DSSs without separate nodes will be analysed.

\subsection{The min-cuts of IFGs with no separate selected nodes}\label{subsec_NoSN}

In this subsection, we assume the selected nodes are all cluster nodes, namely, $\textbf{s}=(s_0=0,s_1,...,s_L)$, in which case, our CSN-DSS model can be seen as cluster DSS model. For given node parameters $(n,k,L,R,S)$, to find the IFG $G^*$ with the minimum min-cut among all possible IFGs is equivalent to find the corresponding selected node distribution $\textbf{s}$ and cluster order $\bm{\pi}$. As \textbf{s} and $\bm{\pi}$ both influence the min-cuts, the analysis comprises two steps:
\begin{adjustwidth}{0.5cm}{0cm}
1. Fix the selected node distribution $\textbf{s}$ and analyse the min-cuts for different cluster orders $\bm{\pi}$ (see the proof of vertical order algorithm in Theorem~\ref{theorem_MC}).\\
2. Fix the cluster order generating algorithm and analyse the min-cuts for different $\textbf{s}$ (see the proof of horizontal selection algorithm in Therorm~\ref{theorem_MC2}).
\end{adjustwidth}

\noindent\textbf{Vertical order algorithm for $d_I=R-1$}

When the selected node distribution $\textbf{s}=(s_0=0,s_1,...,s_L)$ is fixed, the cluster order $\bm{\pi}^*=(\pi_1^*,\pi_2^*,...,\pi_k^*)$ generated by the vertical order algorithm achieves the minimum min-cut among all the possible IFGs, which is proved in Theorem~\ref{theorem_MC}. In~\cite{sohn2017capacity}, the above conclusion is considered based on the special assumption that all the alive nodes are used to repair the failed node, namely, $d_I=R-1$ and $d_C=n-R$. We will investigate and prove this problem in more general settings. The number of helper nodes from cross-cluster, $d_C$, varies from $k-R+1$ to $n-R$ and does not need to be $n-R$ of~\cite{sohn2017capacity}, following from the condition that $d_I+d_C\geq k$.

\begin{algorithm}\label{alg_vs}
\footnotesize
\caption{Vertical order algorithm}
\begin{algorithmic}[1]
\Require $\textbf{s}=(s_0=0,s_1,...,s_L).$ Initial cluster label $j \gets 1;$
\Ensure $\bm{\pi}^*=(\pi_1^*,...,\pi_k^*).$
\For{$i=1$ to k}
    \If {\text{the $i$-th selected node is a separate node}}
         $\pi^*_i\gets 0;$ $continue;$
    \EndIf
    \If{$s_j=0$} $j=1;$
    \Else \ $\pi_i^*\gets j;$ $s_j\gets s_i-1;$ $j\gets (j\mod L)+1;$  \EndIf
\EndFor
\end{algorithmic}
\end{algorithm}

An example of Algorithm 1 is illustrated in Figure~\ref{fig_aiproperty1} (a), where $\textbf{s}=(0,4,3,1)$. After three iterations, $s_3$ equals to $0$ and $\pi_4^*= 1$ in the next iteration. The final output of the algorithm is $\bm{\pi}^*=(1,2,3,1,2,1,2,1)$. If the selected node distribution $\textbf{s}$ is fixed, a cluster order determines an IFG and the min-cut $MC(\textbf{s},\bm{\pi})$ can be calculated. Note that $MC(\textbf{s},\bm{\pi})$ depends on the $i$-th part incoming weight $w_i(\bm{\pi})=a_i(\bm{\pi})\beta_I+b_i(\bm{\pi})\beta_C$ $(1\leq i\leq k)$ and a useful property for $a_i(\bm{\pi})$ (the coefficient of $\beta_I$) is proved in Lemma~\ref{pro_ai}.

\begin{figure*}[!t]
\begin{minipage}[t]{0.49\textwidth}
  \centering
  \subfloat[$\bm{\pi}^*={(1,2,3,1,2,1,2,1)}$]{\includegraphics[width=0.48\textwidth]{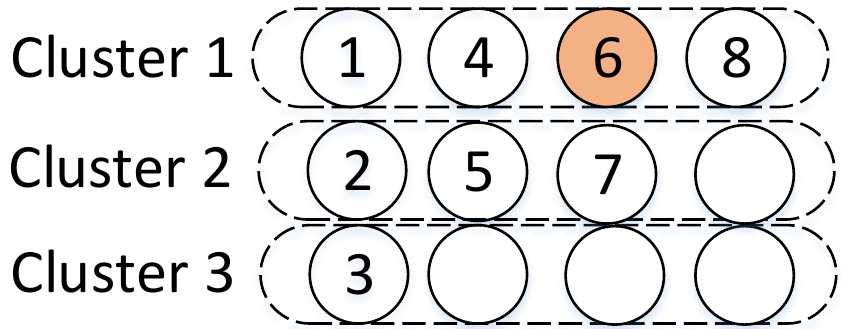}}
  \hspace{0.01\textwidth}
  \subfloat[$\bm{\pi}={(1,2,1,2,1,2,1,3)}$]{\includegraphics[width=0.48\textwidth]{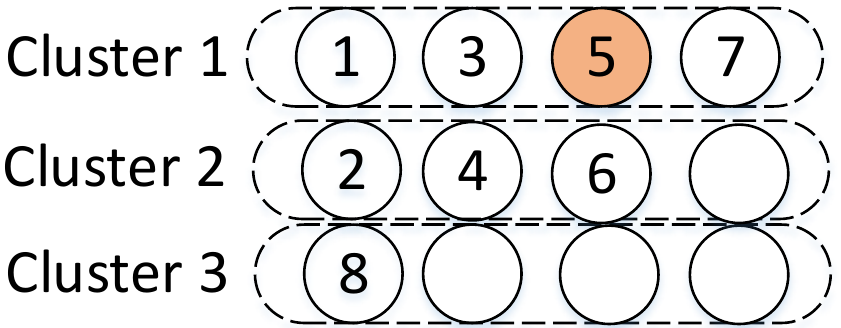}}
  \caption{The numbered nodes are selected nodes. There are two cluster orders $\bm{\pi}^*$ and $\bm{\pi}$ for a selected node distribution $\textbf{s}=(0,4,3,1)$.}\label{fig_aiproperty1}
\end{minipage}
\hspace{0.01\textwidth}
\begin{minipage}[t]{0.49\textwidth}
  \centering
  \subfloat[$\textbf{s}^*=(0,4,4,0)$]{\includegraphics[width=0.48\textwidth]{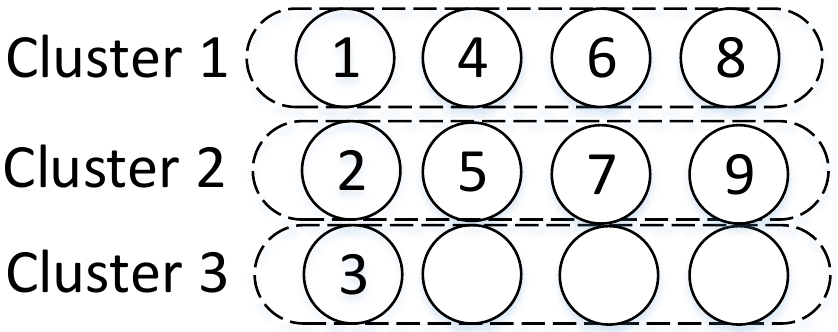}}
  \hspace{0.01\textwidth}
  \subfloat[$\textbf{s}=(0,3,3,2)$]{\includegraphics[width=0.48\textwidth]{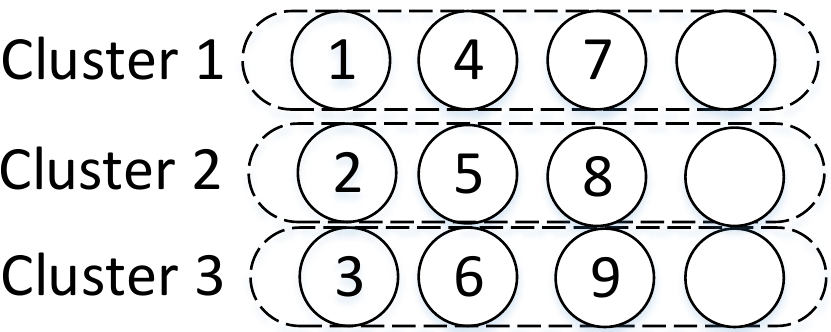}}
  \caption{The numbered nodes are selected nodes. There are two selected node distributions $\textbf{s}^*$ and $\textbf{s}$ for the CSN-DSS $(n=12,k=8,L=3,R=4,S=0)$.}\label{fig_twoSonePi}
\end{minipage}
\end{figure*}

\begin{lemma}\label{pro_ai}
For a given selected node distribution $\textbf{s}=(0,s_1,s_2,...,s_L)$ of the system model in Figure~\ref{fig_CN}, the multi-set\footnote{A multi-set is a generalization of the concept of a set that, unlike a set, allows multiple instances of the multi-set's elements.} $[a_i(\bm{\pi})]_{i=1}^k=[a_1(\bm{\pi}), a_2(\bm{\pi}),..., a_k(\bm{\pi})]$ consists of the same elements for all the different cluster orders $\bm{\pi}\in\Pi(\textbf{s})$ and $a_i(\bm{\pi})=d_I+1-h_{\bm{\pi}}(i)$ for $1\leq i\leq k$.
\end{lemma}

\begin{proof}
Assume $\bm{\pi}^*=(\pi_1^*,..., \pi_k^*)$ and $\bm{\pi}=(\pi_1,...,\pi_k)$ are two different cluster orders for the same selected node distribution $\textbf{s}=(s_1,...,s_L)$. For example, in Figure~\ref{fig_aiproperty1}, the selected node distribution is $\textbf{s}=(4,3,1)$ and the corresponding two different cluster orders are $\bm{\pi}^*=(1,2,3,1,2,1,2,1)$ and $\bm{\pi}=(1,2,1,2,1,2,1,3)$. Note that the coloured node 6 in Figure~\ref{fig_aiproperty1} (a) is the third selected node in Cluster 1. When cutting node 6, two intra-cluster nodes are contained in $\overline{U}$ (node 1 and node 3), which won't be counted in the part-cut value. Then $a_6(\bm{\pi}^*)=d_I-2=R-1-2=1$. Now consider cluster order $\bm{\pi}$ in Figure~\ref{fig_aiproperty1} (b). Although the coloured node 5 is the 5th node in $\bm{\pi}$, it is the third selected node in Cluster 1 and $a_5(\bm{\pi})=d_I-2=R-1-2=1$. It's easy to see that the value of $a_i(\bm{\pi}^*)$ only depends on the number of selected nodes in the same cluster before repairing the current node, which is defined by the relative location $h_{\bm{\pi}^*}(i)$, namely, $a_i(\bm{\pi}^*)=d_I+1-h_{\bm{\pi}^*}(i)$ for $1\leq i\leq k$.

If $\textbf{s}=(0,s_1,s_2,...,s_L)$ is fixed,  the set of $h_{\bm{\pi}}(i)$ for nodes in Cluster $l$ is $\{1, 2,...,s_l\}$ for $1\leq l\leq L$, no matter where the nodes locate in the cluster orders. For all the cluster orders $\pi\in\Pi(\textbf{s})$, the multi-set $[a_i(\pi)]_{i=1}^k$ consists of  $L$ sets $\{d_I+1-1, d_I+1-2,...,d_I+1-s_l\}$ for $1\leq l\leq L$. 
\end{proof}

When the selected node distribution $\textbf{s}$ is fixed, a property of the min-cuts of IFGs for different cluster orders is proved in Theorem~\ref{theorem_MC}.

\begin{theorem}\label{theorem_MC}
For the given node parameters $(n,k,L,R,S)$ and any given selected node distribution $\textbf{s}\in \mathcal{S}$ with $s_0=0$, the vertical cluster order $\bm{\pi}^*$ obtained by the vertical order algorithm achieves the minimum min-cut among all the possible IFGs with $\textbf{s}$. In other words,

$$MC(\textbf{s},\bm{\pi}^*)\leq MC(\textbf{s},\bm{\pi}),$$
holds for arbitrary $\bm{\pi}\in\Pi(\textbf{s})$. $MC(\textbf{s},\bm{\pi})$ is defined by (\ref{equ_MC}).
\end{theorem}

\begin{proof}
Assume $(w_{u_1}(\bm{\pi}),..., w_{u_k}(\bm{\pi}))$ is a non-increasing order of elements in multi-set $[w_i(\bm{\pi})]_{i=1}^k$, namely, $w_{u_1}(\bm{\pi})\geq ...\geq w_{u_k}(\bm{\pi})$. This proof consists of two parts. In Part 1, we will prove that
\begin{small}
\begin{equation}\label{equ_sumw}
\sum_{i=k-t+1}^k w_{u_i}(\bm{\pi}^*) \leq \sum_{i=k-t+1}^k w_{u_i}(\bm{\pi}),
\end{equation}\end{small}
for any $1\leq t \leq k$. Note that (\ref{equ_sumw}) means the sum of the minimum $t$ elements in multi-set $[w_i(\bm{\pi}^*)]_{i=1}^k$ is no more than the sum of the minimum $t$ elements in $[w_i(\bm{\pi})]_{i=1}^k$ for any $1\leq t \leq k$. With the help of (\ref{equ_sumw}), Part 2 completes the proof by considering the relationship between $\alpha$ and $w_i(\bm{\pi}^*)$ in formula (\ref{equ_MC}).

\noindent\textbf{Part 1:} As $w_i(\bm{\pi})= a_i(\bm{\pi})\beta_I+b_i(\bm{\pi})\beta_C$(see equation (\ref{equ_wi})), the coefficient of $\beta_I$ and $\beta_C$ are considered respectively, and let
\begin{equation*}
\phi_i(\bm{\pi})\triangleq a_i(\bm{\pi})+b_i(\bm{\pi}),
\end{equation*}
for simplicity. In the following part, we will compare $\phi_i(\bm{\pi}^*)$ and $\phi_i(\bm{\pi})$ one by one and prove inequality (\ref{equ_phi}) which is important to prove (\ref{equ_sumw}).

For any cluster order $\bm{\pi}\in\Pi(\textbf{s})$, let sequence $(\phi_{t_1}(\bm{\pi}),...,\phi_{t_k}(\bm{\pi}))$ denote the non-increasing order of the elements of multi-set $[\phi_i(\bm{\pi})]_{i=1}^k$, namely, $\phi_{t_1}(\bm{\pi})\geq ...\geq \phi_{t_k}(\bm{\pi}).$ Based on Algorithm 1, it's easy to verify that the sequences $(\phi_1(\bm{\pi}^*),...,\phi_k(\bm{\pi}^*))$, $(b_1(\bm{\pi}^*),...,b_k(\bm{\pi}^*))$ and $(w_1(\bm{\pi}^*),...,w_k(\bm{\pi}^*))$ are all non-increasing. Assume $p$ ($1\leq p\leq k$) is the integer that satisfies the following conditions:
$$b_p(\bm{\pi}^*)=0 \text{ and } b_{p-1}(\bm{\pi}^*)>0,$$
meaning that exactly $d_C$ selected cross-cluster nodes are already cut when cutting the $p$-th node by the cluster order $\bm{\pi}^*$. Then
$\phi_i(\bm{\pi}^*)=a_i(\bm{\pi}^*)+b_i(\bm{\pi}^*)=d-i+1$
for $1\leq i\leq p$. For the remaining nodes in $\bm{\pi}^*$, $b_i(\bm{\pi}^*)=0$ and $\phi_i(\bm{\pi}^*)=a_i(\bm{\pi}^*)$ $(i>p)$.

\noindent $\bullet$ When $1\leq i\leq p$, as $\phi_i(\bm{\pi})\geq d-i+1$ (see (\ref{equ_bi})), $\phi_{t_i}(\bm{\pi})\geq \phi_i(\bm{\pi})\geq d-i+1=\phi_i(\bm{\pi}^*)$.

\noindent $\bullet$ When $p+1\leq i \leq k$, as $a_i(\bm{\pi}^*)\geq a_{i+1}(\bm{\pi}^*)\geq ... \geq a_k(\bm{\pi}^*)$, \textbf{at most $k-i$ elements} of multi-set $[a_i(\bm{\pi}^*)]_{i=1}^k$ are less than $a_i(\bm{\pi}^*)$.

\begin{adjustwidth}{0.5cm}{0cm}
    $\blacktriangleright$ If $a_{t_i}(\bm{\pi}) < a_i(\bm{\pi}^*)$, we first assume $a_{t_j}(\bm{\pi})<a_j(\bm{\pi}^*)$ for all $i+1\leq j\leq k$ and will derive a contradiction. It's obvious that \textbf{at least $k-i+1$ elements} of multi-set $[a_{t_i}(\bm{\pi})]_{i=1}^k$ are less than $a_i(\bm{\pi}^*)$. As is proved in Lemma~\ref{pro_ai}, $[a_{t_i}(\bm{\pi})]_{i=1}^k$ and $[a_i(\bm{\pi})^*]_{i=1}^k$ contain the same elements, $[a_i(\bm{\pi})^*]_{i=1}^k$ then contains \textbf{at least $k-i+1$ elements} of multi-set $[a_{t_i}(\bm{\pi})]_{i=1}^k$ are less than $a_i(\bm{\pi}^*)$, which a contradiction. There then exists at least one $a_{t_j}(\bm{\pi})$ ($i+1\leq j\leq k$) not less than $a_i(\bm{\pi}^*)$, and $\phi_{t_i}(\bm{\pi})\geq \phi_{t_j}(\bm{\pi})\geq a_{t_j}(\bm{\pi})\geq a_{i}(\bm{\pi}^*)= \phi_i(\bm{\pi})$.\\
    $\blacktriangleright$ If $a_{t_i}(\bm{\pi})\geq a_i(\bm{\pi}^*)$, $\phi_{t_i}(\bm{\pi})=a_{t_i}(\bm{\pi})+b_{t_i}(\bm{\pi}) \geq a_{t_i}(\bm{\pi})>a_i(\bm{\pi}^*)=\phi_i(\bm{\pi}^*)$.
\end{adjustwidth}
\noindent Then it can be proved that
\begin{small}
\begin{eqnarray}
\sum_{i=k-t+1}^k\big(a_i(\bm{\pi}^*)+b_i(\bm{\pi}^*)\big)&=&\sum_{i=k-t+1}^k \phi_{i}(\bm{\pi}^*)\leq\sum_{i=k-t+1}^k \phi_{t_i}(\bm{\pi}) \overset{(a)}\leq\sum_{i=k-t+1}^k \big(a_{u_i}(\bm{\pi})+b_{u_i}(\bm{\pi})\big)\label{equ_phi}
\end{eqnarray}
\end{small}
Note that $\sum_{i=k-t+1}^k a_{u_i}(\bm{\pi})+b_{u_i}(\bm{\pi})=\sum_{i=k-t+1}^k\phi_{u_i}(\bm{\pi})$ is the sum of $t$ elements of multi-set $[\phi_i(\bm{\pi})]_{i=1}^k$. Inequality (a) is based on the fact that $\sum_{i=k-t+1}^k \phi_{t_i}(\bm{\pi})$ is the sum of the minimum $t$ elements of $[\phi_i(\bm{\pi})]_{i=1}^k$, not greater than the sum of any $t$ elements of $[\phi_i(\bm{\pi})]_{i=1}^k$.

With the above consequence, it can be proved that
\begin{small}
\begin{eqnarray*}
& &\sum_{i=k-t+1}^k w_{u_i}(\bm{\pi}^*)=\sum_{i=k-t+1}^{k}w_i(\bm{\pi}^*)=\sum_{i=k-t+1}^{k}\left(a_i(\bm{\pi}^*)*\beta_I+b_i(\bm{\pi}^*)*\beta_C\right) \notag\\
&=&\sum_{i=k-t+1}^{k}a_i(\bm{\pi}^*)*\beta_I+\left(\sum_{i=k-t+1}^{k}\left(a_i(\bm{\pi}^*)+b_i(\bm{\pi}^*)\right)-\sum_{i=k-t+1}^{k}a_i(\bm{\pi}^*)\right)*\beta_C \notag\\
&\overset{(b)}{\leq}&\sum_{i=k-t+1}^{k}a_i(\bm{\pi}^*)*\beta_I+\left(\sum_{i=k-t+1}^k \left(a_{u_i}(\bm{\pi})+b_{u_i}(\bm{\pi})\right)-\sum_{i=k-t+1}^{k}a_i(\bm{\pi}^*)\right)*\beta_C \notag\\
&\overset{(c)}{\leq}&\sum_{i=k-t+1}^{k}a_i(\bm{\pi}^*)*\beta_I+\left(\sum_{i=k-t+1}^k a_{u_i}(\bm{\pi}) -\sum_{i=k-t+1}^{k}a_i(\bm{\pi}^*)\right)*\beta_I+\sum_{i=k-t+1}^k b_{u_i}(\bm{\pi})*\beta_C \notag\\
&=&\sum_{i=k-t+1}^k a_{u_i}(\bm{\pi})*\beta_I+\sum_{i=k-t+1}^k b_{u_i}(\bm{\pi})*\beta_C =\sum_{i=k-t+1}^{k}w_{u_{i}}(\bm{\pi}),
\end{eqnarray*}
\end{small}
where $(b)$ is based on inequality~(\ref{equ_phi}) and $(c)$ is because of $\beta_I \geq \beta_C$.

\noindent\textbf{Part 2:}
Assume there are $t_1$ elements in $[w_i(\bm{\pi}^*)]_{i=1}^k$ and $t_2$ elements in $[w_i(\bm{\pi})]_{i=1}^k$ greater than $\alpha$.\\
\begin{adjustwidth}{0.5cm}{0cm}
$\bullet$ If $t_1 < t_2$,
$MC(\textbf{s},\bm{\pi}^*)=t_1\alpha+\sum_{i=t_1+1}^{t_2} w_i(\bm{\pi}^*)+\sum_{i=t_2+1}^{k} w_i(\bm{\pi}^*)\leq t_2\alpha+\sum_{i=t_2+1}^k w_{u_i}(\bm{\pi})=MC(\textbf{s},\bm{\pi})$.\\
$\bullet$ If $t_1=t_2$, it's easy to prove $MC(\textbf{s},\bm{\pi}^*)\leq MC(\textbf{s},\bm{\pi})$, using $(\ref{equ_sumw})$.\\
$\bullet$ If $t_1 > t_2$,
$MC(\textbf{s},\bm{\pi}^*)=t_2\alpha+\sum_{i=t_2+1}^{t_1} \min\{w_i(\bm{\pi}^*), \alpha\}+\sum_{i=t_1+1}^{k} w_i(\bm{\pi}^*)
\leq t_2\alpha+\!\sum_{i=t_2+1}^k w_i(\bm{\pi}^*)\leq t_2\alpha+\sum_{i=t_2+1}^k w_{u_i}(\bm{\pi})=MC(\textbf{s},\bm{\pi})$. 
\end{adjustwidth}
\end{proof}
The consequence of Theorem~\ref{theorem_MC} can be verified by Algorithm 1 and the numerical examples illustrated in Figure~\ref{fig_aiproperty1} (a) and (b).
To analyse the influence of selected node distribution, for any input $\textbf{s}$, let
\begin{equation}\label{equ_pis}
 \bm{\pi}^*(\textbf{s})=(\pi^*(\textbf{s})_1,\pi^*(\textbf{s})_2,...,\pi^*(\textbf{s})_k)
\end{equation}
denote \textbf{the unique cluster order} generated by the vertical order algorithm. In the following part,
we investigate the min-cuts for different selected node distributions $\textbf{s}$, where the cluster orders are $\bm{\pi}^*(\textbf{s})$.

\noindent\textbf{Horizontal selection algorithm for $d_I=R-1$}

The vertical order algorithm generates the cluster order achieving the minimum min-cut for any given selected node distribution $\textbf{s}$. In this part, we assume all the cluster orders are generated by the vertical order algorithm and analyse the min-cuts for different selected node distributions. In Theorem~\ref{theorem_MC2}, it is proved that the minimum min-cut among possible IFGs is achieved by the selected node distribution $\textbf{s}^*=(s_0^*,s_1^*,s_2^*,...,s_L^*)$ generated by the horizontal selection algorithm and the cluster order $\bm{\pi}^*(\textbf{s}^*)$ generated by the vertical order algorithm.

\noindent\textbf{Algorithm 2:} {\emph{Horizontal selection algorithm:}

\noindent\emph{The horizontal selected node distribution} is $\textbf{s}^*=(s_0^*, s_1^*,s_2^*,...,s_L^*)\ (\sum_{i=0}^L s_i^*=k)$, where
\begin{small}
\begin{equation*}
s_i^*=\begin{cases}
               R,                           &i\leq \lfloor \frac{k-s_0^*}{R}\rfloor\\
     k-\lfloor \frac{k-s_0^*}{R}\rfloor R, \ \ \  &i=\lfloor \frac{k-s_0^*}{R}\rfloor+1 \\
               0,                           &i>\lfloor \frac{k-s_0^*}{R}\rfloor+1
    \end{cases}.
\end{equation*}\end{small}

In this section, the situation without separate nodes is considered, namely, $s_0^*=0$. An example of this algorithm is illustrated in Figure~\ref{fig_twoSonePi} (a), where $k=8, R=4$. Based on the horizontal algorithm, $s_1^*=R=4$, $s_2^*=R=4$ and $s_3=k-2R=0$. Another property of $a_i(\bm{\pi})$, the coefficients of $\beta_I$, is proved in the following lemma, when the horizontal selected algorithm is used.

\begin{lemma}\label{pro_ai2}
For the given node parameters $(n,k,L,R,S)$, the coefficients of $\beta_I$ satisfies that
\begin{equation*}
a_i(\bm{\pi}^*(\textbf{s}^*))\leq a_i(\bm{\pi}^*(\textbf{s}))
\end{equation*}
for $1\leq i \leq k$, where $\textbf{s}^*$ is the selected node distribution generated by the horizontal selection algorithm and $\textbf{s}\in \mathcal{S}$ with $s_0=0$. Note that $\bm{\pi}^*(\cdot)$ is defined by (\ref{equ_pis}).
\end{lemma}

\begin{proof}
  As is proved in Lemma~\ref{pro_ai}, the coefficient of $\beta_I$, $a_i(\bm{\pi}^*(\textbf{s}))=d_I+1-h_{\bm{\pi}^*(\textbf{s})}(i)$.  We will analyse the relative location $h_{\bm{\pi}^*(\textbf{s})}(i)$ with jumping points defined in (\ref{equ_jumping}) and  prove that $h_{\bm{\pi}^*(\textbf{s}^*)}(i)\geq h_{\bm{\pi}^*(\textbf{s})}(i)$ for $1\leq i\leq k$.

  Based on the vertical order algorithm, it's easy to verify the following two properties of $h_{\bm{\pi}^*(\textbf{s})}(i)$:

  $\bullet$ $1\leq h_{\bm{\pi}^*(\textbf{s})}(i)\leq R$ for $1\leq i\leq k$,

  $\bullet$ $0\leq h_{\bm{\pi}^*(\textbf{s})}(i+1)-h_{\bm{\pi}^*(\textbf{s})}(i)\leq 1$ for $1\leq i \leq k-1$,

  \noindent meaning that $h_{\bm{\pi}^*(\textbf{s})}(i)$ is non-decreasing and will increase one time at most by $1$. For an example, the sequence of relative location $h_{\bm{\pi}^*(\textbf{s})}(i)$ is $(1,1,1,2,2,3,3,4,4)$ in Figure~\ref{fig_twoSonePi} (a) and when $i=3,5 \text{ or } 7$, the value of $h_{\bm{\pi}^*(\textbf{s})}(i)$ will increase by $1$ for the next time. These values of $i$ are called \textbf{jumping points}, denoted by
  \begin{equation}\label{equ_jumping}
  J(\bm{\pi}^*(\textbf{s}))=(j_0(\bm{\pi}^*(\textbf{s})),j_1(\bm{\pi}^*(\textbf{s})),...,j_{s_1-1}(\bm{\pi}^*(\textbf{s})),j_{s_1}(\bm{\pi}^*(\textbf{s}))),
  \end{equation}
  which depends on $\textbf{s}=(s_0,s_1,...,s_k)$. We set $j_0(\bm{\pi}^*(\textbf{s}))=0$ and $j_{s_1}(\bm{\pi}^*(\textbf{s}))=k$ as the beginning and ending of the jumping point vector, then
  $$j_i(\bm{\pi}^*(\textbf{s}))-j_{i-1}(\bm{\pi}^*(\textbf{s}))=\#\{t| h_{\bm{\pi}^*(\textbf{s})}(t)=i, 1\leq t \leq k\}$$for $1 \leq i\leq s_1$. Based on the definition of cluster order, it's obvious that
  \begin{equation}\label{equ_ji}
    j_i(\bm{\pi}^*(\textbf{s}))-j_{i-1}(\bm{\pi}^*(\textbf{s}))\geq j_{i+1}(\bm{\pi}^*(\textbf{s}))-j_i(\bm{\pi}^*(\textbf{s}))
  \end{equation}
  for $1\leq i\leq s_1-1$.

  We will use induction method to prove $j_i(\bm{\pi}^*(\textbf{s}^*))\leq j_i(\bm{\pi}^*(\textbf{s}))$ for $1\leq i \leq s_1-1$. Based on the vertical order algorithm, $j_1(\bm{\pi}^*(\textbf{s}^*))\leq j_1(\bm{\pi}^*(\textbf{s}))$. Assume $j_t(\bm{\pi}^*(\textbf{s}^*))\leq j_t(\bm{\pi}^*(\textbf{s}))$, it's needed to prove $j_{t+1}(\bm{\pi}^*(\textbf{s}^*))\leq j_{t+1}(\bm{\pi}^*(\textbf{s}))$.

  There are $k-j_{t+1}(\bm{\pi}^*(\textbf{s}^*))$ nodes remaining after jumping point $j_{t+1}(\bm{\pi}^*(\textbf{s}^*))$. Based on the horizontal selection algorithm,
  $$j_{i+1}(\bm{\pi}^*(\textbf{s}^*))-j_i(\bm{\pi}^*(\textbf{s}^*))=j_1(\bm{\pi}^*(\textbf{s}^*)) \text{ or } j_1(\bm{\pi}^*(\textbf{s}^*))-1$$
  for $1\leq i\leq R-1$.
  Then
  \begin{equation}\label{equ_j1}
    k-j_{t+1}(\bm{\pi}^*(\textbf{s}^*))\geq(R-t-1)(j_{t+1}(\bm{\pi}^*(\textbf{s}^*))-j_{t}(\bm{\pi}^*(\textbf{s}^*))-1).
  \end{equation}
  Assume
  \begin{equation}\label{equ_assume}
    j_{t+1}(\bm{\pi}^*(\textbf{s}^*))>j_{t+1}(\bm{\pi}^*(\textbf{s})).
  \end{equation}
  As $j_t(\bm{\pi}^*(\textbf{s}^*))\leq j_t(\bm{\pi}^*(\textbf{s}))$, then
  \begin{eqnarray}
   &&j_{t+1}(\bm{\pi}^*(\textbf{s}))-j_t(\bm{\pi}^*(\textbf{s}))<j_{t+1}(\bm{\pi}^*(\textbf{s}^*))-j_t(\bm{\pi}^*(\textbf{s}^*))\notag \\
   &&\Rightarrow j_{t+1}(\bm{\pi}^*(\textbf{s}))-j_t(\bm{\pi}^*(\textbf{s}))\leq j_{t+1}(\bm{\pi}^*(\textbf{s}^*))-j_t(\bm{\pi}^*(\textbf{s}^*))-1.   \label{equ_j*}
  \end{eqnarray}
  Then
  \begin{eqnarray*}\label{equ_j2}
    k-j_{t+1}(\bm{\pi}^*(\textbf{s}))&&\overset{(a)}{\leq} (s_1-t-1)(j_{t+1}(\bm{\pi}^*(\textbf{s}))-j_{t}(\bm{\pi}^*(\textbf{s}))) \notag\\
                                     &&\overset{(b)}{\leq} (s_1-t-1)(j_{t+1}(\bm{\pi}^*(\textbf{s}^*))-j_{t}(\bm{\pi}^*(\textbf{s}^*))-1)\notag\\
                                     &&\leq (R-t-1)(j_{t+1}(\bm{\pi}^*(\textbf{s}^*))-j_{t}(\bm{\pi}^*(\textbf{s}^*))-1) \notag\\
                                     &&\overset{(c)}{\leq} k-j_{t+1}(\bm{\pi}^*(\textbf{s}^*)),
  \end{eqnarray*}
  where (a) is based on (\ref{equ_ji}), (b) is because of (\ref{equ_j*}) and (c) results from (\ref{equ_j1}). Hence, $$j_{t+1}(\bm{\pi}^*(\textbf{s})^*)\leq j_{t+1}(\bm{\pi}^*(\textbf{s})),$$
  contradicting assumption (\ref{equ_assume}), and it can be proved that
  $j_{t+1}(\bm{\pi}^*(\textbf{s}^*))\leq j_{t+1}(\bm{\pi}^*(\textbf{s})).$ Since $h_{\bm{\pi}^*(\textbf{s})}(i)=t$ for $j_{t-1}(\bm{\pi}^*(\textbf{s}))\leq i \leq j_{t}(\bm{\pi}^*(\textbf{s}))$ $(t=1,2,...,s_1)$ and $j_i(\bm{\pi}^*(\textbf{s}^*))\leq j_i(\bm{\pi}^*(\textbf{s}))$ for $1\leq i \leq s_1-1$, it can be proved that
  \begin{equation}\label{equ_hpis}
    h_{\bm{\pi}^*(\textbf{s}^*)}(i)\geq h_{\bm{\pi}^*(\textbf{s})}(i),
  \end{equation} for $1\leq i\leq k$. Hence, $a_i(\bm{\pi}^*(\textbf{s}^*))\leq a_i(\bm{\pi}^*(\textbf{s}))$ for $1\leq i \leq k$.
\end{proof}

\begin{theorem}\label{theorem_MC2}
For the given node parameters $(n,k,L,R,S)$, when the selected node distribution $\textbf{s}^*$ is generated by the horizontal selection algorithm and the corresponding cluster order is generated by the vertical order algorithm, the min-cut of this IFG isn't greater than any IFGs. In other words,

\begin{equation*}
  MC(\textbf{s}^*, \bm{\pi}^*(\textbf{s}^*))\leq MC(\textbf{s}, \bm{\pi}^*(\textbf{s})),
\end{equation*}
for all $\textbf{s}\in \mathcal{S}$ with $s_0=0$. Note that $\bm{\pi}^*(\cdot)$ is defined by (\ref{equ_pis}). $MC(\textbf{s},\bm{\pi})$ is defined by (\ref{equ_MC}).
\end{theorem}

\begin{proof}
Based on the vertical order algorithm, sequence $(w_1(\bm{\pi}^*(\textbf{s})), ..., w_k(\bm{\pi}^*(\textbf{s})))$ is non-increasing for all $\textbf{s}\in \mathcal{S}$ and $s_0=0$. Similarly to the proof of Theorem~\ref{theorem_MC}, it is only needed to prove that
\begin{equation*}
w_i(\bm{\pi}^*(\textbf{s}^*))\leq w_i(\bm{\pi}^*(\textbf{s}))
\end{equation*}
for $1\leq i \leq k$. From the definition of $a_i(\bm{\pi}^*(\textbf{s}))$, $b_i(\bm{\pi}^*(\textbf{s}))$ and $h_{\bm{\pi}}(i)$(see (\ref{equ_hi})), it is known that
\begin{equation}\label{equ_aipi}
a_i(\bm{\pi}^*(\textbf{s}))=d_I+1-h_{\bm{\pi}^*(\textbf{s})}(i)
\end{equation}
for $1\leq i\leq k$.
When $d_C-(i-h_{\bm{\pi}^*(\textbf{s})}(i))\geq 0$,
\begin{equation}\label{equ_bipi}
b_i(\bm{\pi}^*(\textbf{s}))=d_C-(i-h_{\bm{\pi}^*(\textbf{s})}(i)).
\end{equation}
We assume $b_i(\bm{\pi}^*(\textbf{s}))$ decreases to $0$ when $i=i^*(\textbf{s})$, where $i^*(\textbf{s})$ is a function of $\textbf{s}$. It will not decrease anymore and
$w_i(\bm{\pi}^*(\textbf{s}))= a_i(\bm{\pi}^*(\textbf{s}))\beta_I$ for $i\geq i^*(\textbf{s})$, where
\begin{equation}\label{equ_is}
i^*(\textbf{s})-h_{\bm{\pi}^*(\textbf{s})}(i^*(\textbf{s}))=d_C.
\end{equation}
As is proved in Lemma~\ref{pro_ai2} (\ref{equ_hpis}), $h_{\bm{\pi}^*(\textbf{s})}(i) \leq h_{\bm{\pi}^*(\textbf{s}^*)}(i)$ for $1\leq i\leq k$,
then $i^*(\textbf{s}^*)\geq i^*(\textbf{s})$ based on (\ref{equ_is}).

\noindent $\bullet$ When $1\leq i\leq i^*(\textbf{s})$,
\begin{small}
\begin{eqnarray*}
  w_i(\bm{\pi}^*(\textbf{s}))-w_i(\bm{\pi}^*(\textbf{s}^*))&=&\big(a_i(\bm{\pi}^*(\textbf{s}))-a_i(\bm{\pi}^*(\textbf{s}^*))\big)\beta_I+ \big(b_i(\bm{\pi}^*(\textbf{s}))-b_i(\bm{\pi}^*(\textbf{s}^*))\big)\beta_C \notag\\
  &\overset{(a)}{=}&\big(h_{\bm{\pi}^*(\textbf{s}^*)}(i)-h_{\bm{\pi}^*(\textbf{s})}(i)\big)\beta_I-\big(h_{\bm{\pi}^*(\textbf{s}^*)}(i)-h_{\bm{\pi}^*(\textbf{s})}(i)\big)\beta_C \\
  &=&\big(h_{\bm{\pi}^*(\textbf{s}^*)}(i)-h_{\bm{\pi}^*(\textbf{s})}(i)\big)(\beta_I-\beta_C)
  \overset{(b)}{\geq} 0.
\end{eqnarray*}
\end{small}
Note that $(a)$ is based on (\ref{equ_aipi}) and (\ref{equ_bipi}). $(b)$ comes from (\ref{equ_hpis}) and $\beta_I\geq\beta_C$. Then $w_i(\bm{\pi}^*(\textbf{s}))\geq w_i(\bm{\pi}^*(\textbf{s}^*)$.

\noindent $\bullet$ When $i^*(\textbf{s})+1\leq i\leq i^*(\textbf{s}^*)$, $b_i(\bm{\pi}^*(\textbf{s}))$ equals to $0$ and will not decrease with $i$ increasing, but
\begin{equation}\label{equ_bil}
  d_C-(i-h_{\bm{\pi}^*(\textbf{s})}(i))\leq 0.
\end{equation}
Hence,
\begin{small}
\begin{eqnarray*}
w_i(\bm{\pi}^*(\textbf{s}))-w_i(\bm{\pi}^*(\textbf{s}^*))
&=&\big(a_i(\bm{\pi}^*(\textbf{s}))-a_i(\bm{\pi}^*(\textbf{s}^*))\big)\beta_I-b_i(\bm{\pi}^*(\textbf{s}^*)\beta_C \notag\\
&\overset{(c)}{\geq}&\big(h_{\bm{\pi}^*(\textbf{s}^*)}(i)-h_{\bm{\pi}^*(\textbf{s})}(i)\big)\beta_I-b_i(\bm{\pi}^*(\textbf{s}^*)\beta_C-(d_C-(i-h_{\bm{\pi}^*(\textbf{s})}(i)))\beta_C\\
&=&\big(h_{\bm{\pi}^*(\textbf{s}^*)}(i)-h_{\bm{\pi}^*(\textbf{s})}(i)\big)\beta_I-\big(h_{\bm{\pi}^*(\textbf{s}^*)}(i)-h_{\bm{\pi}^*(\textbf{s})}(i)\big)\beta_C\\
&=&\big(h_{\bm{\pi}^*(\textbf{s}^*)}(i)-h_{\bm{\pi}^*(\textbf{s})}(i)\big)(\beta_I-\beta_C)\geq 0,
\end{eqnarray*}
\end{small}
where $(c)$ bases on (\ref{equ_bil}).

\noindent $\bullet$ When $i^*(\textbf{s}^*)+1\leq i\leq k$,
$w_i(\bm{\pi}^*(\textbf{s}))= a_i(\bm{\pi}^*(\textbf{s}))\beta_I\geq a_i(\bm{\pi}^*(\textbf{s}^*))\beta_I = w_i(\bm{\pi}^*(\textbf{s}^*))$.
\end{proof}
The consequence of Theorem~\ref{theorem_MC2} can be verified by Algorithm 2 and the numerical examples illustrated in Figure~\ref{fig_twoSonePi} (a) and (b). As is proved by Theorem~\ref{theorem_MC} and Theorem~\ref{theorem_MC2}, the capacity of a cluster DSS with given node parameters and storage/repair parameters is $MC(\textbf{s}^*, \bm{\pi}^*(\textbf{s}^*))$. Based on (\ref{equ_bound}), the tradeoff between node storage and repair bandwidth can be characterized, which is illustrated in Section \ref{sec_consruction} for specific numerical parameters.

\subsection{The Min-cuts of IFGs with separate selected nodes} \label{subsec_WithSN}

As is proved in last subsection, the horizontal selection algorithm and the vertical order algorithm will generate the selected node distribution $\textbf{s}^*$ and the corresponding cluster order $\bm{\pi}^*(\textbf{s}^*)$, achieving the minimum min-cut of the IFGs without separate selected nodes. In this subsection, we analyse the min-cut of IFGs with one separate selected nodes in Theorem~\ref{theorem_MC3}, namely $s_0=1$ in $\textbf{s}$.

\begin{figure*}[!t]
  \centering
  \subfloat[$\bm{\pi}^*=(1,2,0,1,2,1,2,1)$]
  {\includegraphics[width=0.25\textwidth]{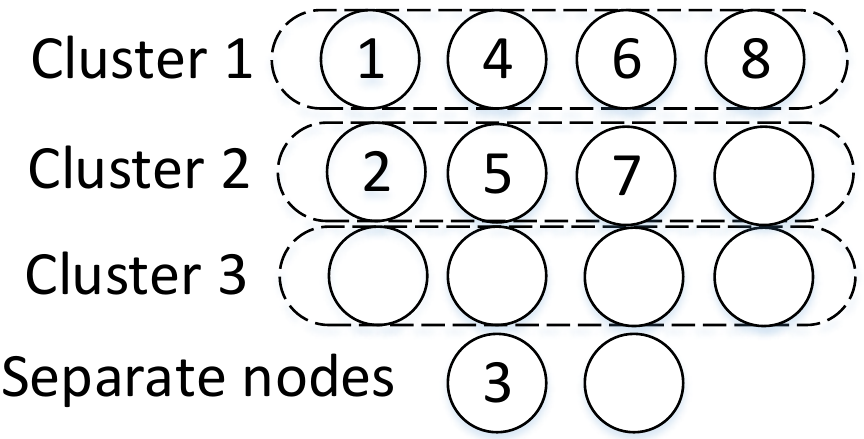}}
  \hspace{0.02\textwidth}
  \subfloat[$\bm{\pi}=(1,2,0,1,2,1,2,3)$]
  {\includegraphics[width=0.25\textwidth]{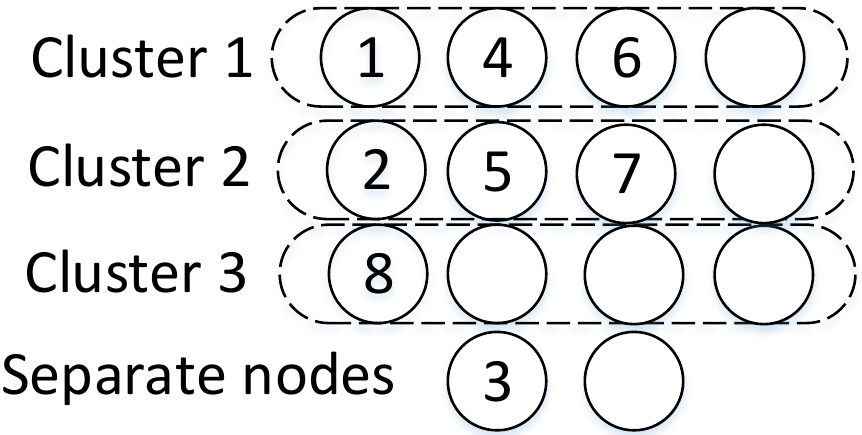}}
  \hspace{0.02\textwidth}
  \subfloat[$\overline{\bm{\pi}}=(1,2,1,2,1,2,3,1)$]
  {\includegraphics[width=0.25\textwidth]{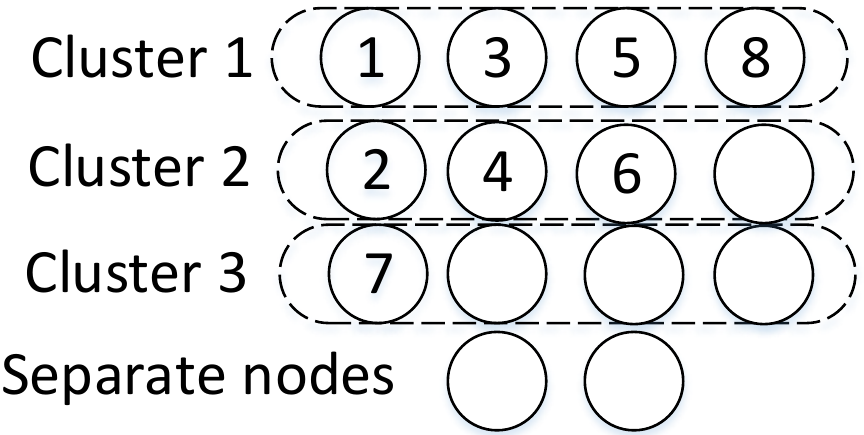}}
  \caption{The numbered nodes are selected nodes. There are three selected node distributions $\textbf{s}^*=(1,4,3,0)$, $\textbf{s}=(1,3,3,1)$ and $\overline{\textbf{s}}=(0,4,3,1)$, with cluster orders $\bm{\pi}^*$, $\bm{\pi}$ and $\overline{\bm{\pi}}$ respectively for CSN-DSS with node parameters $(n=12,k=9,L=3,R=4,S=2)$.}\label{fig_SC}
\end{figure*}

Note that Theorem~\ref{theorem_MC} focuses on the influence of cluster order $\bm{\pi}$ where the selected node distribution $\textbf{s}$ is fixed. On the other hand, Theorem~\ref{theorem_MC2} analyses different selected node distributions while the cluster order generating algorithm is fixed. The following theorem combines these two aspects and investigates the situation where there is one separate selected node.

\begin{theorem}\label{theorem_MC3}
  For the given node parameters $(n,k,L,R,S)$, assume the $j$-th selected node is a separate node, then the selected node distribution $\textbf{s}^*$ generated by the horizontal selection algorithm with cluster order $\bm{\pi}^*(\textbf{s}^*)$ generated by the vertical order algorithm minimizes the min-cut of all the IFGs. In other words,
  \begin{equation*}
    MC(\textbf{s}^*, \bm{\pi}^*(\textbf{s}^*))\leq MC(\textbf{s}, \bm{\pi}),
  \end{equation*}
  for all $\textbf{s}\in \mathcal{S}$ with $s_0=1$ and $\bm{\pi}\in \Pi(\textbf{s})$ with $\pi_j=0$. Note that $\bm{\pi}^*(\cdot)$ is defined by (\ref{equ_pis}). $MC(\textbf{s},\bm{\pi})$ is defined by (\ref{equ_MC}).
\end{theorem}

Due to space limitation, here we just sketch the proof of Theorem~\ref{theorem_MC3} in the following part.

\noindent \textbf{Cluster order assignment:} Let $\overline{\bm{\pi}}=(\overline{\pi}_1,\overline{\pi}_1,...,\overline{\pi}_k)$ denote a cluster order with no separate selected nodes in its corresponding selected node distribution $\overline{\textbf{s}}$. When the $j$-th selected node is a separate node, let $\bm{\pi}$ denote the new cluster order with
\begin{small}
\begin{equation*}\label{equ_newpi}
  \pi_i=\begin{cases}
          \overline{\pi}_i, & \mbox{if } 1\leq i< j \\
          0, & \mbox{if } i=j \\
          \overline{\pi}_{i-1}, & \mbox{if } j<i\leq k
        \end{cases}.
\end{equation*}\end{small}
As is illustrated in Figure~\ref{fig_SC} (b) (c), node 3 is a separate selected node in $\bm{\pi}$, then $\pi_1=\overline{\pi}_1$, $\pi_2=\overline{\pi}_2$, $\pi_3=0$ and $\pi_i=\overline{\pi}_{i-1}(i=4,5,6,7,8)$. In Figure~\ref{fig_SC} (a), the optimal selected node distribution is $\textbf{s}^*=(1,4,3,0)$ and the optimal cluster order is $\bm{\pi^*}=(1,2,0,1,2,1,2,1)$ generated by the vertical order algorithm. Note that the 3rd node in the cluster order is a separate node, which is fixed beforehand.

\noindent \textbf{Main idea of the proof:} For any cluster order $\bm{\pi}$ whose $j$-th node is a separate node, there always exists a cluster order $\overline{\bm{\pi}}$ with no separate selected nodes satisfying the definition of $\bm{\pi}$. Through investigating the relationship between $\overline{\bm{\pi}}$ and $\bm{\pi}$, the proof of Theorem~\ref{theorem_MC3} can be reduced to similar problems of Theorem~\ref{theorem_MC} and Theorem~\ref{theorem_MC2}, which is omitted here.

To analyse the minimum min-cut of IFGs with one separate selected node, the location of the separate selected node (the value of $j$ in Theorem~\ref{theorem_MC3}) is also important. Moreover, the relation among $\beta_I$, $\beta_C$ and $\beta_S$ need to be taken into consideration. The situation with multiple separate selected nodes can be investigated using similar methods, which is not introduced here due to the space limitation.

\section{Numerical Results and Code Constructions for Cluster DSSs}\label{sec_consruction}

In this section, Figure~\ref{fig_bounds} illuminates some numerical capacity bounds for cluster DSSs without separate nodes. As is mentioned in \cite{Survey2011,Fazeli2016}, interference alignment is an important method in regenerating code constructions, which is also applicative in cluster DSSs. A code construction strategy with interference alignment is investigated for a cluster DSS with specific parameters as an example(see Figure~\ref{fig_46bound} and Figure~\ref{fig_consFig}). The code constructions for general cases can utilize similar methods.

\begin{figure*}
 \begin{minipage}[t]{0.35\textwidth}
  \centering
  \includegraphics[width=\textwidth]{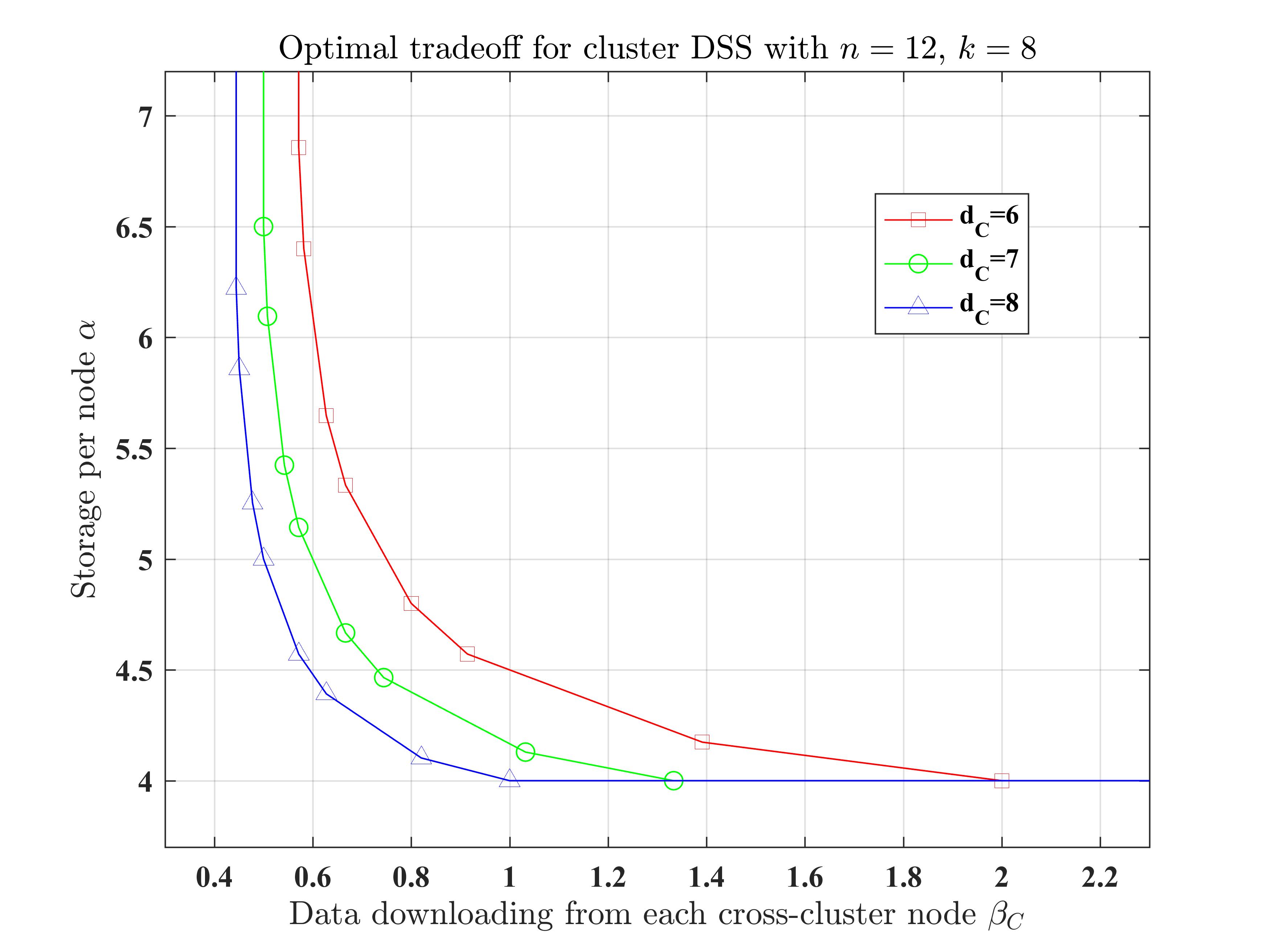}
  \caption{Optimal tradeoff curves between storage per node $\alpha$ and data downloading from each cross-cluster node $\beta_C$, for cluster DSS with $n=12$, $k=8$ and $\beta_I=2\beta_C$. There are three curves for different number of cross-cluster helper nodes, namely, $d_C=6,7,8$, respectively. The original file size is $\mathcal{M}=32$.}\label{fig_bounds}
\end{minipage}
\hspace{0.01\textwidth}
\begin{minipage}[t]{0.35\textwidth}
  \centering
  \includegraphics[width=\textwidth]{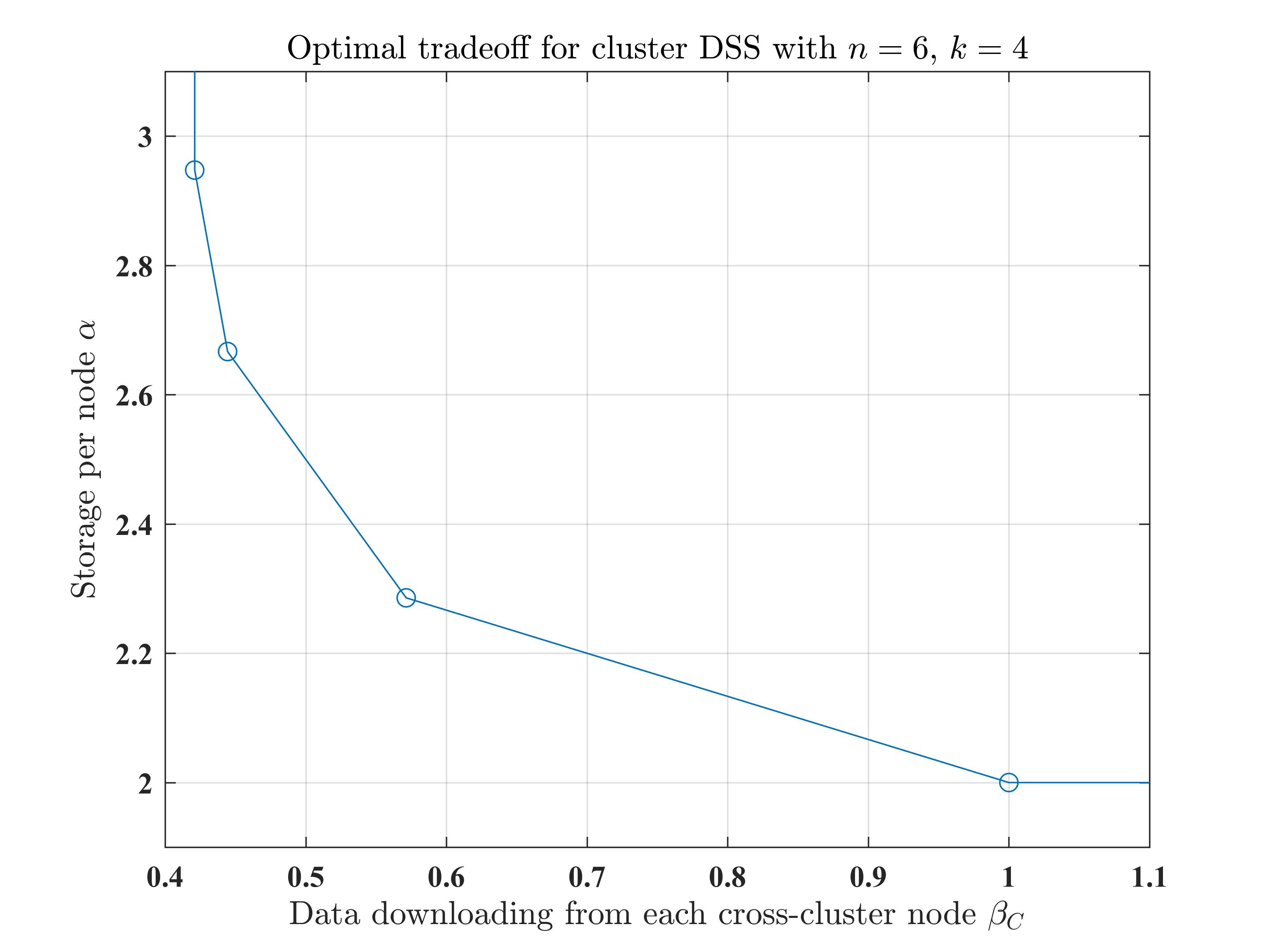}
  \caption{Optimal tradeoff curves between storage per node $\alpha$ and data downloading from each cross-cluster node $\beta_C$, for cluster DSS with $n=6$, $k=4$ and $\beta_I=2\beta_C, d_C=3$. The original file size is $\mathcal{M}=8$.}\label{fig_46bound}
\end{minipage}
\hspace{0.01\textwidth}
\begin{minipage}[t]{0.25\textwidth}
  \centering
  \includegraphics[width=\textwidth]{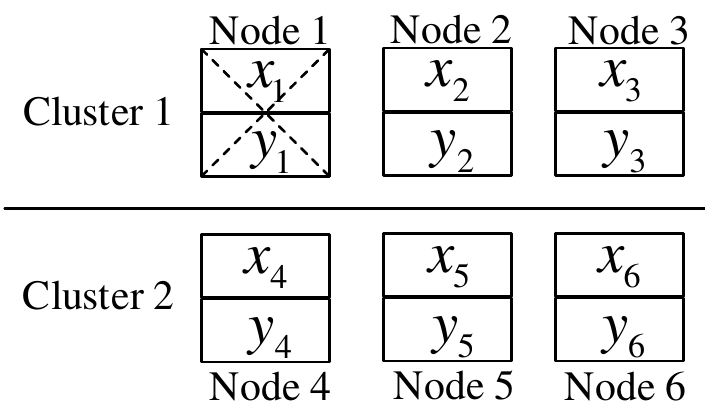}
  \caption{MSR construction for cluster DSS with parameters $n=6,k=4,\beta_I=2\beta_C, d_C=3$ and $\mathcal{M}=8$.}
  \label{fig_consFig}
 \end{minipage}
\end{figure*}

\noindent \textbf{Model configurations:} Assume the node parameters are $(n=12, k=8, L=3, R=4, S=0)$ as Figure~\ref{fig_twoSonePi} shows. Note that different relations between $\beta_I$ and $\beta_C$ result in different tradeoffs between storage per node $\alpha$ and bandwidth to repair one node. We consider a specific situation that $\beta_I=2\beta_C$. As is proved before, the cluster order illustrated in Figure~\ref{fig_twoSonePi} (a) achieves the capacity of this DSS. Based on the bound constraint in (\ref{equ_bound}), for specific values of $k-R+1\leq d_C\leq n-k$, the tradeoff between $\alpha$ and $\beta_C$ is illuminated in Figure~\ref{fig_bounds}, where the file size $\mathcal{M}$ is set to be $32$ for simplicity.

As is illuminated in Figure~\ref{fig_bounds}, the tradeoff curve moves left as the cross-cluster helper nodes increase. Since $\alpha$ is the storage per node and $\beta_C$ corresponds to the repair bandwidth, when the storage per node is fixed, the more helper nodes are utilized, the less bandwidth will be, which is consistent with the consequences of DSS without clusters in \cite{Dimakis2010}. When $d_C=8$, the point $(\alpha=4, \beta_C=2)$ achieves the minimum storage and the corresponding code constructions is called minimum-storage regenerating (MSR) codes, as is defined in \cite{Dimakis2010}. We will investigate an MSR construction for cluster DSS with less nodes as an example.

Consider the cluster DSS with $n=6, k=4$ in Figure~\ref{fig_consFig}. It can be verified in Figure~\ref{fig_46bound} that the MSR point of this system is $(\alpha=2, \beta_C=1)$, hence the amount of data downloading from each intra-cluster node is $\beta_I=2\beta_C=2$.

\noindent \textbf{Encoding procedure:} As the original file size is $\mathcal{M}=8$, we assume $x_i(1\leq i\leq 4)$ and $y_i(1\leq i\leq 4)$ are the original file symbols storing from Node~1 to Node~4 as Figure~\ref{fig_consFig} illustrates. Two (6,4)-MDS codes are used to encode symbols $(x_1,x_2,x_3,x_4)$ and $(y_1,y_2,y_3,y_4)$, respectively. Let
\begin{small}\begin{equation}
(x_1,x_2,x_3,x_4,x_5,x_6)=(x_1,x_2,x_3,x_4)\left[\mathbf{I}_{4\times 4}|\mathbf{g}\ \mathbf{h}\right] \text{ and }(y_1,y_2,y_3,y_4,y_5,y_6)=(y_1,y_2,y_3,y_4)\left[\mathbf{I}_{4\times 4}|\mathbf{g}'\ \mathbf{h}'\right],\end{equation}\end{small}
where $\mathbf{I}_{4\times 4}$ is an identity matrix and $\mathbf{g}=(g_1,g_2,g_3,g_4)^t$, $\mathbf{h}=(h_1,h_2,h_3,h_4)^t$, $\mathbf{g'}=(g'_1,g'_2,g'_3,g'_4)^t$, $\mathbf{h}=(h'_1,h'_2,h'_3,h'_4)^t$.
Then
\begin{small}
\begin{eqnarray*}
x_5&=&g_1x_1+g_2x_2+g_3x_3+g_4x_4,\ y_5=g'_1y_1+g'_2y_2+g'_3y_3+g'_4y_4,\\
x_6&=&h_1x_1+h_2x_2+h_3x_3+h_4x_4,\ y_6=h'_1y_1+h'_2y_2+h'_3y_3+h'_4y_4.
\end{eqnarray*}
\end{small}

\noindent \textbf{Repair procedure:} Assume Node 1 has failed, based on the tradeoff in Figure~\ref{fig_46bound} ($\alpha=2$ and $\beta_C=1,\ \beta_I=2\beta_C=2$), we will download $2$ symbols each from Node 2 and Node 3 respectively and download $1$ symbols each from Node~4 to Node~6. As the values of $x_2,y_2,x_3,y_3$ is known by downloading from Node 2 and Node 3, to calculate $x_1$ and $y_1$, interference alignment can be used to eliminate $x_4$ and $y_4$. For example, the symbols downloading from Node 4, Node 5 and Node 6 are
\begin{small}
\begin{eqnarray*}
symbol_4&=&l_4x_4+l'_4y_4,\\
symbol_5&=&m_5x_5+m'_5y_5=m_5(g_1x_1+g_2x_2+g_3x_3+g_4x_4)+m'_5(g'_1y_1+g'_2y_2+g'_3y_3+g'_4y_4),\\
symbol_6&=&n_6x_6+n'_6y_6=n_6(h_1x_1+h_2x_2+h_3x_3+h_4x_4)+n'_6(h'_1y_1+h'_2y_2+h'_3y_3+h'_4y_4),
\end{eqnarray*}
\end{small}
respectively.
Hence,
\begin{small}
\begin{eqnarray}
symbol_4&=&l_4x_4+l'_4y_4, \label{equ_node4}\\
symbol_5-m_5(g_2x_2+g_3x_3)-m'_5(g'_2y_2+g'_3y_3)&=&m_5g_1x_1+m'_5g'_1y_1+m_5g_4x_4+m'_5g'_4y_4,\label{equ_node5}\\
symbol_6-n_6(h_2x_2+h_3x_3)-n'_6(h'_2y_2+h'_3y_3)&=&n_6h_1x_1+n'_6h'_1y_1+n_6h_4x_4+n'_6h'_4y_4.\label{equ_node6}
\end{eqnarray}\end{small}
Note that the left parts of Eq.~\ref{equ_node4}, Eq.~\ref{equ_node5} and Eq.~\ref{equ_node6} are real values. If the coefficients of $x_1$, $y_1$ and $x_4$, $y_4$ satisfy that
\begin{small}
\begin{equation}\label{equ_rank}
\textbf{rank}\left(\left[\begin{matrix}
             m_5g_1 & m'_5g'_1 \\
             n_6h_1 & n'_6h'_1
           \end{matrix}\right]\right)=2 \text{ and }
\textbf{rank}\left(\left[\begin{matrix}
                l_4 & l'_4    \\
             m_5g_4 & m'_5g'_4 \\
             n_6h_4 & n'_6h'_4
           \end{matrix}\right]\right)=1,
\end{equation}
\end{small}
$x_4$ and $y_4$ can be eliminated, meanwhile, $x_1$ and $y_1$ can be solved, finishing the repair of Node 1.

As is proved in \cite{Survey2011}, there exists MDS codes satisfying the condition (\ref{equ_rank}). When other nodes have failed, similar methods can be utilized to generate the parameters of corresponding MDS codes. The construction of MDS codes adaptive to all the node failures and cluster DSS with general parameters is more complicated and need more future work.

\section{Conclusion and future work} \label{sec_conclusion}

In this paper, DSSs with clusters and separate nodes are investigated, where the tradeoff between node storage and repair bandwidth is characterized on more flexible parameter settings. When a node in cluster DSSs has failed, the number of helper nodes varies based on the practical storage system demands. The influence of separate nodes is also analysed for DSSs with clusters and separated nodes. Moreover, a regenerating code construction strategy is proposed for cluster DSSs with specific parameters as a numerical example, achieving the points in the optimal tradeoff curve.

More general and practical regenerating codes for cluster DSSs need to be investigated further. On the other hand,
more works are needed to characterize the influence of separate selected nodes for the min-cuts of CSN-DSSs more explicitly, when analysing the optimal tradeoff between node storage and repair bandwidth. Furthermore, the constructions of more flexible regenerating codes for CSN-DSSs are also meaningful for practical storage systems.



\bibliographystyle{IEEEtran}
\bibliography{DSbib}
%
%
%

\end{document}